\documentclass[preprint,12pt]{elsarticle}
\usepackage[utf8]{inputenc}
\usepackage{amsthm}
\usepackage{amsmath}
\usepackage{amsfonts}
\usepackage{amssymb}
\usepackage{graphicx}
\usepackage{epstopdf}
\usepackage{color}
\usepackage{multirow}
\usepackage{mathtools}
\usepackage{xcolor}
\usepackage{a4wide}
\usepackage{bm}
\usepackage{hyperref}
\usepackage{caption}
\usepackage{subcaption}
\usepackage[
  separate-uncertainty = true,
  multi-part-units = repeat
]{siunitx}
\usepackage{rotating}
\usepackage{verbatim}

\usepackage{tikz}
\tikzstyle{noeud} = [circle, draw, fill=white, inner sep=2pt]
\newtheorem{theorem}{Theorem}[section]
\newtheorem{lemma}{Lemma}[section]

\newdefinition{remark}{Remark}[section]
\newtheorem{observation}{Observation}[section]
\newdefinition{definition}{Definition} \newdefinition{example}{Example}

\usepackage[noend,linesnumbered,ruled,lined]{algorithm2e}
\SetAlFnt{\scriptsize}
\usepackage{amsfonts}
\usepackage{makecell}

\usepackage[usenames,dvipsnames]{pstricks}
\usepackage{pstricks-add}
\usepackage{epsfig}
\usepackage{pst-grad} 
\usepackage{pst-plot} 
\usepackage[space]{grffile} 
\usepackage{etoolbox} 
\makeatletter 
\patchcmd\Gread@eps{\@inputcheck#1 }{\@inputcheck"#1"\relax}{}{}
\makeatother
\usepackage{mathtools}

\usepackage{graphicx}
\usepackage{tabularx}
\usepackage{textcomp}
\usepackage{xcolor}
\usepackage[labelformat=simple]{subcaption}

\usepackage{booktabs}

\usepackage{pgf,tikz,pgfplots}
\pgfplotsset{compat=1.15}
\usepackage{mathrsfs}
\usepackage{hyperref}
\usepackage{multicol}
\SetAlFnt{\scriptsize}
\usepackage{amsfonts}
\usepackage{makecell}
\usepackage{bm}

\usepackage{multirow}
\usepackage{array}
\usepackage{makecell}
\usepackage{graphicx}
\newcolumntype{P}[1]{>{\raggedright\arraybackslash}p{#1}}

\SetCommentSty{mycommfont}

\usepackage{a4wide}
\begin{document}
\usetikzlibrary{arrows}
\begin{frontmatter}

\title{Black Hole Search in Dynamic Graphs}

            


\author[inst1]{Tanvir Kaur}
\ead{tanvir.20maz0001@iitrpr.ac.in}

\author[inst1]{Ashish Saxena}
\ead{ashish.21maz0004@iitrpr.ac.in}

\author[inst2]{Partha Sarathi Mandal}
\ead{psm@iitg.ac.in}

\author[inst1]{Kaushik Mondal}
\ead{kaushik.mondal@iitrpr.ac.in}

\affiliation[inst1]{organization={Department of Mathematics},
            addressline={Indian Institute of Technology Ropar}, 
            city={Rupnagar},
            postcode={140001}, 
            state={Punjab},
            country={India}}
            
\affiliation[inst2]{organization={Department of Mathematics},
            addressline={Indian Institute of Technology Guwahati}, 
            city={Guwahati},
            postcode={781039}, 
            state={Assam},
            country={India}}


 \begin{abstract}
    A black hole is considered to be a dangerous node present in a graph that disposes of any resources that enter that node. Therefore, it is essential to find such a node in the graph. Let a group of agents be present on a graph $G$. The Black Hole Search (BHS) problem aims for at least one agent to survive and terminate after {finding} the black hole. This problem is already studied for specific dynamic graph classes such as rings, cactuses, and tori {where finding the black hole means at least one agent needs to survive and terminate after knowing at least one edge associated with the black hole. In this work, we investigate the problem of BHS for general graphs.} In the dynamic graph, adversary may remove edges at each round keeping the graph connected. We consider two cases: (a) at any round at most one edge can be removed (b) at any round at most $f$ edges can be removed. For both scenarios, we study the problem when the agents start from a rooted initial configuration. We consider each agent has $O(\log n)$ memory and each node has $O(\log n)$ storage.

    \noindent For case (a), we present an algorithm with $9$ agents that solves the problem of BHS in $O(|E|^2)$ time where $|E|$ is the number of edges and $\delta_v$ is the degree of the node $v$ in $G$. We show it is impossible to solve for $2\delta_{BH}$ many agents starting from an arbitrary configuration where $\delta_{BH}$ is the degree of the black hole in $G$. We also provide another improved algorithm that uses $6$ agents from a rooted initial configuration to solve the problem of BHS.

    \noindent For case (b), we provide an algorithm using $6f$ agents to solve the problem of BHS, albeit taking exponential time. We also provide an impossibility result for $2f+1$ agents starting from a rooted initial configuration. This result holds even if unlimited storage is available on each node and the agents have infinite memory.
\end{abstract}

\end{frontmatter}

\section{Introduction}

In distributed computing, a substantial literature has delved into the computational paradigm involving mobile agents. The employment of mobile agents to address diverse global challenges through distributed means stands as a widely recognized approach in problem-solving methodologies. A mobile agent, in this context, refers to a software entity endowed with limited memory, capable of traversing a network and visiting nodes. Upon visiting a node, it engages in computation, interacting with the local environment, the memory, and the resources of the visited node. Subsequently, the agent proceeds to the next node based on the computations performed. In the mobile agent paradigm, many problems are studied, such as exploration \cite{Albers_1997, dynamic_exploration},   gathering \cite{DPP14, LunaTCS2020}, dispersion \cite{Augustine_2018, dynamic_dispersion}, etc.

Practically, inconsistency and faults are unavoidable in the network. By faults, we mean faulty agents, faulty nodes, or even dangerous graphs. From the dangerous graphs perspective, a black hole is one such issue that points to nodes where the resources get affected or destroyed. As an example, in the case of mobile computing, if an agent reaches a black hole, it may die. From the practical side, a black hole can occur when a computer is accidentally offline or a resident process (such as a virus) on a network site deletes incoming data or agents without leaving a trace. For example, in a cloud, a node that causes loss of essential data can compromise the quality of any service in that cloud and become a black hole. Similarly, any undetectable crash failure of a site in a network can turn that site into a black hole. It is therefore important to find such a node. This problem is known as the black hole search problem. It involves a team of agents whose task is to find the black hole.

The problem of black hole search is well studied in the context of static graphs \cite{Paola_2006, Paola_2004}. Usually, the underlying graph is considered to be anonymous, i.e., the nodes have no ID. However, there is a growing interest in exploring this problem in dynamic graphs, where edges can be removed as long as the underlying graph remains connected. Previous studies have concentrated on the black hole search problem in dynamic rings \cite{dyn_ring_jrnl} and dynamic cactus graphs \cite{bhattacharya_2023}. Recently, the problem is studied for the case of dynamic tori graph\cite{Adri_tori}. However, there is no research on the black hole search problem on arbitrary dynamic graphs. Our objective in this paper is to extend the study of the black hole search problem to arbitrary graphs and analyze the problem for cases when a single edge can be deleted by the adversary and when at most $f$ edges can be deleted by the adversary in each discrete time step from the graph.

\subsection{Related Work}

The black hole search problem was initially presented by Dobrev et al. \cite{Paola_2006} where they addressed the problem for a static, arbitrary topology. The problem of black hole search is extensively studied on static graphs \cite{Paola_2006, Pelc_2005, Paola_2010, Shantanu_2011, MarkouP12, complexity_black_hole, bhs_tokens}. Here in this section, we focus on the existing literature of black hole search on dynamic graphs.


 Di Luna et al. \cite{dyn_ring_jrnl} first introduced the black hole search problem on dynamic rings. Their objective is as follows. At least one agent needs to survive and terminate after knowing at least one edge associated with the black hole. In this work, the authors established lower bounds and presented size-optimal algorithms, considering factors such as the number of agents, moves, and rounds. These results are based on two communication models: face-to-face and whiteboard. They show that three agents are necessary to find the black hole in the ring when the adversary can delete at most one edge from the ring in each round. When the three agents are colocated initially, they can solve the BHS problem in the whiteboard and pebble communication model in $\Theta(n^{1.5})$ rounds even if the agents are anonymous. If agents can communicate only when they are on the same node, the authors show that the black hole search problem can be solved in $\Theta(n^2)$ moves and rounds if the agents have unique IDs and is impossible to solve if the agents are anonymous. Further, the authors investigate the case when agents begin from an arbitrary initial configuration in rings. They provide a time optimal algorithm that finds the black hole in $\Theta(n^2)$ rounds in the pebble model. They also show that it is impossible to solve BHS problem from arbitrary initial configuration only using face-to-face communication model.  

Bhattacharya et al. \cite{bhattacharya_2023} presented two algorithms for black hole search in cactus graphs. In each round, if at most one edge is deleted from the footprint $G$ provided TVG $\mathcal{G}$ remains connected, then they established a lower bound of $\Omega(n^{1.5})$ rounds and an upper bound of $O(n^2)$ rounds with the help of three agents. With four agents, a lower bound of $\Omega(n)$ rounds was provided. For cases where multiple edges can be deleted by the adversary (up to $f$ edges provided the graph remains connected), they provided the impossibility of {finding} the black hole with $f+1$ agents. However, with $f+2$ agents, they established a lower bound of $\Omega((n+2-3f)^{1.5})$ rounds. The authors in \cite{bhattacharya_2023} use $O(n\cdot\log n + n\cdot f\cdot\log f)$ storage on the whiteboard. Recently, Bhattacharya et al. \cite{Adri_tori} investigated the BHS problem in dynamic tori for both rooted and arbitrary initial configurations. In their work, the dynamicity is such that each ring is considered to be 1-interval connected, i.e., at most one edge can be missing from each ring at any round. Each node of the torus has a unique identifier. Each node has a whiteboard with constant bits of storage for each of the associated edges of that node. They provided an algorithm that solves BHS using $n+3$ agents in $O(n\cdot m^{1.5})$ rounds, where the size of tori $n\times m$ ($3\leq n \leq m$). 
\begin{table*}[h!]
\centering
\resizebox{\textwidth}{!}{%
\begin{tabular}{!{\vrule width 2pt}P{2.3cm}|P{0.6cm}|P{1.4cm}|P{1.6cm}|P{1.6cm}|P{1.7cm}|P{2.7cm}|P{1.6cm}|P{3.8cm}!{\vrule width 2pt}}
\noalign{\hrule height 2pt}
\textbf{Graph} & \textbf{IC} & \textbf{Prob.} & \textbf{\makecell{Agents\\ know\\ missing\\ edge}} & \textbf{Comm.} & \textbf{Min Agents Reqd.} & \textbf{Time} & \textbf{No. of Agents} & \textbf{Memory (Whiteboard)} \\
\noalign{\hrule height 2pt}

\multirow{4}{*}{Ring \cite{dyn_ring_jrnl}} & \multirow{2}{*}{R} & \multirow{4}{*}{1-BHS} & Yes & F2F & 3 & $\Theta(n^2)$ & 3 & - \\ \cline{4-9}
& & & Yes & Wb & 3 & $\Theta(n^{1.5})$ & 3 & $O(\log n)$ \\ \cline{2-2} \cline{4-9}
& \multirow{2}{*}{A} & & Yes & F2F & 3 & Impossible & 3 & - \\ \cline{4-9}
& & & Yes & Pebble & 3 & $\Theta(n^2)$ & 3 & - \\ 
\noalign{\hrule height 2pt}

\multirow{2}{*}{Cactus \cite{bhattacharya_2023}} & \multirow{2}{*}{R} & 1-BHS & Yes & \multirow{2}{*}{\makecell{Wb \&\\F2F}} & 3 & $O(n^2)$ & 3 & \multirow{2}{*}{$O(\delta_v(\log\delta_v+f\log f))$} \\ \cline{3-4} \cline{6-8}
& & $f$-BHS & Yes & & $f+2$ & $O(fn)$ & $2f+3$ & \\
\noalign{\hrule height 2pt}

\multirow{4}{*}{\makecell{Tori \cite{Adri_tori}\\(Size $n \times m$)}} & \multirow{2}{*}{R} & \multirow{4}{*}{\makecell{$(n+m)$-\\BHS}} & Yes & \multirow{4}{*}{\makecell{Wb \&\\F2F}} & $n+2$ & $O(nm^{1.5})$ & $n+3$ & \multirow{4}{*}{$O(\delta_v)$} \\ \cline{4-4} \cline{6-8}
& & & Yes & & & $O(nm)$ & $n+4$ & \\ \cline{2-2} \cline{4-4} \cline{6-8}
& \multirow{2}{*}{A} & & Yes & & $n+3$ & $O(nm^{1.5})$ & $n+6$ & \\ \cline{4-4} \cline{6-8}
& & & Yes & & & $O(nm)$ & $n+7$ & \\
\noalign{\hrule height 2pt}

\multirow{5}{*}{\makecell{General\\Graph\\(Our Results)}} & \multirow{3}{*}{R} & 1-BHS & No & \multirow{5}{*}{\makecell{Wb \&\\F2F}} & - & $O(|E|^2)$ & 9 &  \\ \cline{3-4} \cline{6-8}
& & $1$-BHS & No & & $-$ & $O(|E|^2)$ & $6$ & $O(\log n)$\\ \cline{3-4} \cline{6-8}
& & $f$-BHS & No & & $2f+2$ & $3\Delta^n(\Delta+1)^{2f+n}(n-1)^{2f}$ & $6f$ & \\ \cline{2-4} \cline{6-9}
& \multirow{2}{*}{A} & 1-BHS & No & & $2\delta_{BH}+1$ & - & - & $O(\log n)$ \\ \cline{3-4} \cline{6-9}
& & $f$-BHS & No & & $2f+2$ & - & - & Infinite \\
\noalign{\hrule height 2pt}
\end{tabular}
}
\caption{The table provides a summary of the existing results and our results. In col. 2, IC denotes initial configuration, and R, A denote rooted and arbitrary respectively. In col. 5, F2F and Wb denote face-to-face communication and whiteboard model respectively.}
\label{fig:literature}
\end{table*}
Further, using $n+4$ agents, they provided an algorithm that solves BHS in dynamic tori in $O(m\cdot n)$ rounds. When the agents start from an arbitrary initial configuration, they provide two algorithms. The first one uses $n+6$ agents to solve the BHS problem in $O(n\cdot m^{1.5})$ rounds. The second one uses $n+7$ agents to solve the BHS problem in $O(n\cdot m)$ rounds.

The problem of black hole search on arbitrary graphs is still an open question. In this paper, we investigate the problem of BHS on an arbitrary graph. In our model, we use $O(\log n)$ storage on the whiteboard instead of $O(\delta_v(\log\delta_v+k\cdot\log k))$ storage per node used in  \cite{bhattacharya_2023} for the rooted case. Furthermore, Bhattacharya et al.\cite{bhattacharya_2023} assume that each agent knows the missing edges at its current node at any round, whereas we consider that an agent can only determine if an edge is missing or not by attempting to move through it. The Table \ref{fig:literature} depicts a summary of existing literature and our results.

\subsection{Model and problem definition} 

\noindent \textbf{Dynamic graph model:} The system is represented as a time-varying graph (TVG), denoted by $\mathcal{G} = (V, E, \mathbb{T}, \rho)$, where, $V$ represents a set of nodes, $E$ represents a set of edges, $\mathbb{T}$ represents the temporal domain, which is defined to be $\mathbb{Z}^+$ as in this model we consider discrete time steps, and $\rho: E \times \mathbb{T} \rightarrow \{0, 1\}$ tells whether a particular edge is available at a given time. The graph $G = (V, \,E)$ is the underlying static graph of the TVG $\mathcal{G}$, also termed as the footprint of $\mathcal{G}$, where the number of nodes is $|V|=n$ and $|E|$ denotes the number of edges in $G$. We denote $\delta_v$ as the degree of node $v$ in the footprint $G$, and $\Delta$ as the maximum degree in the footprint $G$. The graph $G$ is undirected and unweighted. The nodes of $G$ are anonymous (i.e., they have no IDs), and each node $v$ is equipped with a whiteboard of size $O(\log n)$. The graph $G$ is port-labeled, i.e., each edge incident on a node $v$ is assigned a locally unique port number in $\{0, 1, 2, 3,\cdots, \delta_v-1\}$. Specifically, each undirected edge $(u,\,v)$ has two labels, denoted as its port numbers at $u$ and $v$. Port numbering is local, i.e., there is no relationship between the port numbers at $u$ and $v$. In graph $G$, there is a black hole. If any agent visits such a dangerous node, then it disposes off the agent without leaving any trace of its existence. We denote the black hole node by $BH$ and its degree by $\delta_{BH}$, however, agents do not have any prior knowledge about the black hole. 
We call the remaining nodes of the graph $G$ safe nodes. In our dynamic graph model, the adversary has the ability to remove edge(s) from $G$ at any particular time step, provided the graph remains connected at each round.

\vspace{0.2cm}
\noindent\textbf{Agent model}: We consider $k$ agents to be present initially at safe nodes of the graph $G$. Each agent has a unique identifier assigned from the range $[1, n^c]$, where $c > 1$ is some constant unknown to the agents. Each agent knows its ID. Agents do not have any prior knowledge of $k$, or any other graph parameters like, $n$, $\delta_{BH}$, $\Delta$. The agents are equipped with $O(\log n)$ memory. They have access to the whiteboard and have the capability to read, write, or delete the written information on it. When an agent visits a node, it knows the degree of that node in the footprint $G$. However, at a node $v$, an agent can not sense if any edge at $v$ is missing. To be precise, say an agent $a$ reaches a node $v$ at round $t$ and $e_v$ be an edge associated with $v$, then the agent $a$ can not determine $\rho(e_v,\,t)$. It can only understand this by traversing through the edge. If it tries to move through an edge and is unable to do so, it understands that the edge is missing. In other words, it understands that $\rho(e_v,\,t)=0$ at the beginning of the next round $t+1$. Otherwise ($\rho(e_v,\,t)=1$), it can successfully move through that edge. An agent learns the port number through which it enters a node. If two agents move via an edge, they can not see each other during the movement.

\vspace{0.2cm}
Our algorithm runs in synchronous rounds. In each round, an agent performs one \textit{Look-Compute-Move} (LCM) cycle at a safe node:
\begin{itemize}
    \item Look: The agent reads the contents of the whiteboard of the node it occupies and sees if there are other agents at the same node. Each agent can read the parameters of other agents, which are present at the same node. The agent also understands whether it had a successful or an unsuccessful move in the last round.
    \item Compute: On the basis of the information obtained in the Look phase, the agent decides whether to move through the port or not in this round. The agent may write some information on the whiteboard as per the algorithm.
    \item Move: In compute, if an agent computes some port $p$ to move, it tries to move through the port $p$; if the corresponding edge is present in that round, the agent reaches the adjacent node; otherwise, it remains at the current node.
\end{itemize}
The time complexity of the algorithm is defined as the number of rounds until an agent finds the black hole in the graph $G$.

\vspace{0.1cm}

\begin{definition}
    \textbf{$(f$-BHS$)$} Let $\mathcal{G}$ be a TVG. There is a black hole present in the footprint $G$. Let $k$ agents be located arbitrarily at safe nodes of $G$. The agents have no information about the black hole and missing edges in $G$. Agents do not know the parameters $n$, $k$, $\Delta$ and $\delta_{BH}$. At each round, the adversary can remove at most $f$ edges from the graph $G$, provided the graph $\mathcal{G}$ remains connected. At least one agent needs to {find} the black hole in the footprint $G$ of $\mathcal{G}$ and terminate.
\end{definition}
\noindent If $f=1$ in the $f$-BHS problem, we denote it as \textbf{$1$-BHS} problem.

\vspace{0.1cm}

Our problem definition is the same as the black hole search problem on dynamic rings \cite {dyn_ring_jrnl}, dynamic cactuses \cite{bhattacharya_2023} and dynamic tori \cite{Adri_tori}, where at least one agent needs to survive and terminate after knowing at least one edge associated with the black hole. Here we study on arbitrary graphs and provide algorithms starting from rooted initial configurations. In the next section, we discuss our contribution in detail.

\subsection{Our contribution}
We study both $1$-BHS and $f$-BHS from a rooted initial configuration on a much weaker model compared to the ones used in \cite{dyn_ring_jrnl, bhattacharya_2023, Adri_tori}. Our model is weaker for the following reason. In our model, the agents cannot detect whether an edge associated with its current node $v$ in $G$ is missing or not in any particular round $t$ while doing look, compute, and move in that round $t$. The agent can understand whether the movement was successful (the edge was present) or unsuccessful (the edge was missing) during the look of the round $t+1$. However, according to the model of \cite{dyn_ring_jrnl, bhattacharya_2023, Adri_tori}, an agent while doing look at round $t$ being on a node $v$, understands either no associated edge to the node $v$ is missing in this round or the exact port number corresponding to which one edge associated to the node $v$ is missing. In other words, in  \cite{dyn_ring_jrnl, bhattacharya_2023, Adri_tori}, the agent computes on a locally static structure with respect to its current position, and hence the move in each round is always successful.

\vspace{0.2cm}
\noindent We obtain the following results for the $1$-BHS problem.

\begin{itemize}
    \item  We prove the impossibility of solving $1$-BHS with $2\delta_{BH}$ many agents arbitrarily placed on safe nodes of $G$ (arbitrary configuration), provided that the agents have $O(\log n)$ memory and the nodes have a whiteboard of $O(\log n)$ storage. (Refer to Section \ref{sec:imp_1bhs}).
    \item  We provide an algorithm to solve $1$-BHS with $9$ co-located agents at a safe node of $G$ (rooted configuration) in $O(|E|^2)$ time. The agents require $O(\log n)$ memory, and each node is equipped with a whiteboard of storage $O(\log n)$. (Refer to Section \ref{sec:1-dynamic}).
    \item  We provide an improved algorithm with respect to number of agents by solving $1$-BHS with only $6$ co-located agents at a safe node of $G$ (rooted configuration) in $O(|E|^2)$ time. The agents require $O(\log n)$ memory, and each node is equipped with a whiteboard of storage $O(\log n)$. (Refer to Section \ref{sec:algo_6agents}).
\end{itemize}

\noindent We obtain the following results for the $f$-BHS problem.
\begin{itemize}
    \item  We prove the impossibility of solving $f$-BHS with $2f+1$ agents co-located at a safe node of $G$ (rooted configuration). The result holds even if the agents have infinite memory and the nodes have infinite storage. Since rooted configuration is a special case of arbitrary configuration, this impossibility holds for arbitrary configuration as well. (Refer to Section \ref{sec:imp_fbhs}).
    \item  We provide an algorithm that solves $f$-BHS from a rooted configuration using $6f$ agents with $O(\log n)$ memory with each agent. Each node of the graph is equipped with a whiteboard of storage $O(\log n)$. The time taken by the agents to solve BHS is $3\Delta^n(\Delta+1)^{2f+n}(n-1)^{2f}$ rounds. (Refer to Section \ref{sec:f dynamic}).
\end{itemize}

 In Section \ref{Exp:1}, we provide an algorithm with 3 co-located agents in $G$ that solves exploration in $O(|E|^2)$ rounds, assuming the adversary removes at most one edge from $G$ in each round. In Section \ref{EXP:2}, we show that the exploration can also be accomplished with just 2 co-located agents within the same asymptotic round complexity under the same adversarial condition.

A summary of our contribution to solve BHS is given in the Table \ref{fig:literature}.

\subsection{Preliminaries}\label{sec:pre}
The exploration problem involves ensuring that every node in a graph is visited by one or more mobile computational entities known as agents. Visits can occur a finite number of times (referred to as exploration with termination) or can happen infinitely (known as perpetual exploration). If a group of agents can explore a dynamic graph, then the same exploration strategy can also be helpful in identifying an unsafe node within the graph. Hence, we begin with an exploration strategy over a graph that does not have any black hole.

In Section \ref{sec:1-dynamic}, we employ a technique based on Depth First Search (DFS) traversal to solve the $1$-BHS problem. For completeness, we discuss this standard DFS-based exploration technique: DFS consists of two fundamental states: $explore$ and $backtrack$. Initially, the agent is in the $explore$ state. In this state, when the agent is at node $v$, it selects an unexplored edge (if available) and traverses it to a new node. If the agent arrives at a node that has already been visited, or if there are no unexplored edges remaining at node $v$, the agent then transitions to the $backtrack$ state and returns to the previous node in its traversal path. This process continues recursively until all reachable nodes have been explored. The following theorem is on the runtime for exploring a static graph using DFS.

\begin{theorem}\label{Thm:DFS}
    \cite{Das_2019} Any static graph $G$ with $|E|$ many edges can be explored by an agent within $4|E|$ rounds using DFS traversal.
\end{theorem}

In \cite{Nicolo_21}, Gotoh et al. consider perpetual exploration by mobile agents on dynamic graphs whose topology is arbitrary and unknown to the agents. They focus on the solvability of the exploration of such dynamic graphs and specifically on the number of agents that are necessary and sufficient for exploration. We use this to solve $f$-BHS in Section \ref{sec:f dynamic}. Authors in \cite{Nicolo_21} use certain notations, which are as follows: $\mathcal{G}$ is $l$-bounded 1-interval connected if it is always connected and $|\overline{E_t}| \leq l$. Let $\mathcal{W}(l)$ denote the class of $l$-bounded 1-interval connected temporal graphs. Here $\overline{E_t} \;(\subseteq E)$ is the set of edges that do not appear at time $t$.

Gotoh et al. \cite{Nicolo_21} give an algorithm $\mathcal{EXPD}$ which solves the exploration problem for $l$-bounded 1-connected graphs. In the algorithm, one of the agents is designated as the leader, while the rest are the non-leader agents. Each agent knows whether it is the leader or not. We use the following theorem from \cite{Nicolo_21}  in Section \ref{sec:f dynamic}.

\begin{theorem} \label{thm:EXP_algo} \cite{Nicolo_21}
1-interval connected time-varying graphs $\mathcal{G} \in \mathcal{W}(l)$ is explored by  $2l$ agents with a leader within $\Delta^{n} (\Delta +1)^{2l+n} (n-1)^{2l}$ rounds, where $\Delta$ is the largest degree of $\mathcal{G}$. 
\end{theorem}

\section{Impossibility for $1$-BHS}\label{sec:imp_1bhs}
In this section, we give a construction of a footprint $G$ of a TVG $\mathcal{G}$ where it is impossible to {find} the black hole in the presence of $2 \delta_{BH}$ many agents, where $\delta_{BH}$ is the degree of the black hole node.  In our construction, we consider an arbitrary initial configuration of the agents. Note that, despite this result, the possibility of solving the problem with constant many agents starting from a rooted initial configuration remains open.
We start with the following lemma on map construction of a clique of a certain size. Map construction requires the agents to compute a copy of the graph including all port labels \cite{Das_2019}.

\begin{lemma}\label{lm:Imp} 
If $Cl$ is a clique of size $\sqrt{n}$, and each node (denoted as $v$) of $Cl$ has storage of $O(\log n)$, then two agents with $O(\log n)$ memory cannot find the map of $Cl$.
\end{lemma}
\begin{proof}
    If the size of clique $Cl$ is $\sqrt{n}$, then the number of edges in $Cl$ is $n$. To find the map of $Cl$, we need $n \cdot \log n$ storage. In $Cl$, each node has $O(\log n)$ storage as a whiteboard, and each agent has $O(\log n)$ memory. Therefore, the total available storage in $Cl$ (including agents memory) is $O(\sqrt{n} \cdot \log n)$. This storage can not store $n \cdot \log n$ much information because $\sqrt{n} \cdot \log n=o(n \cdot \log n)$. Therefore, agents can not store the map of $Cl$. 
\end{proof}

The idea for the construction of $G$ for our impossibility result relies on Lemma \ref{lm:Imp}. If there is an edge (say $e$) between a node of clique $Cl$ and the black hole node, then to understand the fact that the adversary can not remove edge $e$, agents need to understand the map of $Cl$. Our impossibility result is as follows.

\begin{theorem}\label{thm:imp-1-BHS}
    It is impossible for $2\delta_{BH}$ agents to solve $1$-BHS problem with agents having $O(\log n)$ memory and nodes having $O(\log n)$ storage.    
    
\end{theorem}
\begin{proof}
To show our impossibility result, we provide a graph construction where $2\delta_{BH}$ agents can not solve the $1$-BHS problem. 
Let $G$ be a graph of size $n$ and $v_{BH}$ be a black hole in $G$. Consider $m_1=\lfloor \sqrt{n} \rfloor$ and $m_2= \lfloor \frac{n-1}{m_1} \rfloor$. Let the degree of $v_{BH}$ be $m_1$. Let $Cl_1$, $Cl_2$, \ldots, $Cl_{m_1}$ be the cliques of size $m_2$. In $G$, clique $Cl_i$ is connected with node $v_{BH}$ via edge $e_i$ only. There is no edge between any two distinct cliques $Cl_i$ and $Cl_j$. We can see one such example in Fig. \ref{fig:Imp:1}. We consider two agents present in each $Cl_i$. Therefore, $2 \delta_{BH}$ many agents are present in $G$.

\begin{figure}
    \centering
    \includegraphics[width=0.4\linewidth]{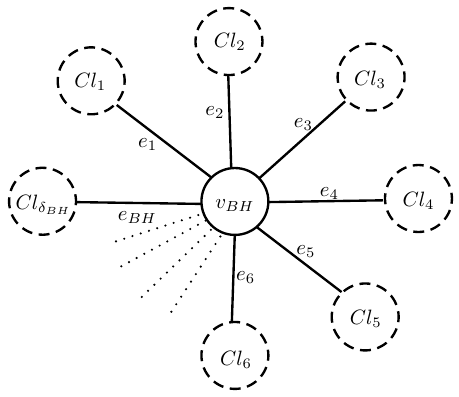}
    \caption{Construction of the class of graph for impossibility of $1-$BHS}
    \label{fig:Imp:1}
\end{figure}


In the above configuration, we see that each edge $e_i$ is a bridge connecting $Cl_i$ and black hole. The adversary cannot delete any of these edges as it disconnects the graph. But the agents can understand this only if the agents present in $Cl_i$ can find the map of the sub-graph $Cl_i$ as it sees $e_i$ as the only edge that connects  $Cl_i$ to the rest of the graph.  Due to Lemma \ref{lm:Imp}, agents in clique $Cl_i$ can not find the map of $Cl_i$, where $i \in [1, \delta_{BH}]$ as in this configuration, the size of the clique is $m_2$, where $m_2 = O(\sqrt{n})$.

In order to {find} the black hole, one of the agents must go through edge $e_i$ in the bid to reach the other part of the graph. If agent $a$ from $Cl_i$ goes through edge $e_i$, it gets destroyed at $v_{BH}$. In subsequent rounds, if there is any agent $b$ in $Cl_i$, it either waits for the agent $a$ indefinitely or also gets destroyed at node $v_{BH}$. So, the possibilities are, both the agents wait forever in clique $Cl_i$, or the agent $a$ dies at $v_{BH}$ and $b
$ waits forever in clique $Cl_i$ or both dies at $v_{BH}$. Therefore, in our graph construction, $2\delta_{BH}$ many agents are insufficient to solve the $1$-BHS problem.
\end{proof}


\section{$1$-BHS from Rooted Configuration using 9 Agents}\label{sec:1-dynamic}
In this section, we solve the 1-BHS problem with 9 agents starting from a rooted initial configuration where the adversary can remove at most one edge from $G$ in each round.
If the graph is static, our objective can be solved using 2 agents from a rooted configuration by doing a DFS traversal but in a little different manner. A round $t$ of the DFS needs to be executed in three sub-rounds, say $t_1$, $t_2$, $t_3$. Let two agents, say $a_1$ and $a_2$, be at node $v$ at the start of round $t$, i.e., sub-round $t_1$. In $t_1$, agent $a_1$ follows the round $t$ of the DFS algorithm, and the other agent stays at node $v$. In $t_2$ sub-round, $a_1$ returns to the node $v$. In the $t_3$ sub-round, both agents follow the round $r$ of the DFS algorithm. If, at some round, $a_1$ does not return to its position, the agent $a_2$ gets to know the black hole. This type of movement of agents is known as \emph{cautious walk} \cite{dyn_ring_jrnl, bhattacharya_2023, Adri_tori}. However, we cannot use this method to solve our problem as the adversary can remove an edge from $G$ at the beginning of a round. 

\vspace{0.2cm}
\noindent\textbf{Challenges:} The adversary can remove at most one edge from $G$ at the beginning of a round. When this occurs, agent(s) must wait at the node until the missing edge reappears, to continue the movement according to the DFS traversal. And if they continue to wait, the algorithm may get stuck forever. So the challenges include coming up with a modified exploration technique that works in our model, determining the number of agents required to run that algorithm, and determining the waiting time, if any, once agents understand a missing edge. Also, in the existing solution of cactus \cite{bhattacharya_2023} and ring \cite{dyn_ring_jrnl}, agents can learn which edge is missing during look phase, whereas in our model, the agents cannot learn missing edges during look. This poses a challenge in itself.

\vspace{0.2cm}
\noindent\textbf{Our approach:} We provide an exploration strategy using $3$ agents starting from rooted initial configuration (Section \ref{Exp:1}). As usual, exploration means that every node of $G$ needs to be visited by at least one agent. Then we present a strategy for cautious movement by a group of agents (Section \ref{Sec: Cautious_walk}). Finally, the agents find the black hole using the exploration strategy where edge movements are conducted by cautious walk (Section \ref{sec: Algorithm}).

\subsection{Exploration strategy with correctness}\label{Exp:1}
A high-level strategy for exploration, when the adversary can remove an edge from $G$ at the start of each round, is as follows.

\vspace{0.2cm}
\noindent\textbf{High-level idea:} We consider three agents, say $a_1$, $a_2$, and $a_3$. Without loss of generality, let $a_1.ID<a_2.ID<a_3.ID$, where $a_i.ID$ represents the ID of an agent $a_i$ for each $i \in \{1,2,3\}$. All agents execute DFS from the same node, say $v$. The basic idea of exploration is that when an adversary removes an edge $e=(u,\,v)$, the smallest ID agent $a_1$ stays at the same node and waits for the appearance of the edge. The remaining agents (if any) restart the execution of the depth-first search or skip the edge to continue the ongoing depth-first search algorithm. All the agents execute their depth-first search algorithm in this manner. Let an edge $e'=(u',\, v')$ be removed at round $t$. At round $t$, one agent is at node $u'$ while another is at $v'$. When the third agent reaches $u'$ or $v'$, the minimum ID agent waits for $e'$ to reappear while the other proceeds with its DFS. Hence, at a time, one missing edge can affect no more than $2$ agents; thus, the $3^{rd}$ agent continues performing its DFS movement in that round. In other words, at least one agent performs its DFS movement at every round. Now we proceed with a detailed description of the algorithm for exploration.

\vspace{0.2cm}
\noindent\textbf{Description of algorithm:} Three agents, namely, $a_1$, $a_2$, $a_3$, stars from a node, say $v_1$. Each agent $a_i$ maintains the following parameters:
\begin{itemize}
    \item $a_i.ID$: It denotes the ID of agent $a_i$.
    \item $a_i.success$: It denotes whether the movement of agent $a_i$ is successful. In other words, when an agent $a_i$ attempts to move from a node $u$ to $v$ via the edge $e$ in round $t$, it is either successful or unsuccessful. If the agent reaches the node $v$ in the round $t+1$, it is a successful attempt, thus, $a_i.success=True$; otherwise, $a_i.success=False$. It is initially set to $True$. 
    \item $a_i.label$: It denotes how many times agent $a_i$ has started the DFS since the start of the algorithm. It is initially set to $-1$. 
    \item $a_i.state$: It denotes the state of its state in the current DFS. Agent $a_i$ can have either an $explore$ or a $backtrack$ state. It is initially set to $explore$. 
    \item $a_i.pin$: It stores the port through which an agent $a_i$ enters into the current node. It is initially set to $-1$.  
    \item $a_i.pout$: It stores the port through which an agent $a_i$ exits from the current node. It is initially set to $-1$.
    \item $a_i.r$: It stores values either $0$ or $1$. If the current round $t$ is such that $t\mod2=0$, then $a_i.r=0$. Else, it is set to $1$.
\end{itemize}

The nodes of the graph are equipped with a whiteboard. We represent the contents of the whiteboard as follows. The whiteboard at a node $v$ is represented as $wb_v$ and information corresponding to each agent $a_i$ is stored on it. Precisely, $wb_v(a_i.parent, a_i.label)$ represents that the agent $a_i$ first visited the node $v$ via the port $a_i.parent$ and this information corresponds to the $(a_i.label)^{th}$ DFS being run by the agent $a_i$.  

Initially, $a_i.state=explore$, and $a_i.success=True$ for each $i$. In the round $0$, each agent $a_i$ updates their $parent$ as $-1$ and $label$ as $1$ on the whiteboard of the node $v_1$. In other words, $wb_{v_1}(a_i.parent, a_i.label)=(-1,1)$. They set $a_i.pout=0$ and move through the port $0$. If the edge corresponding to port $0$ is present, then the agents are successful in reaching a new node, else they remain at the same node. The decision of an agent $a_i$ being at a node $w$ at round $t$ is described below in detail based on different cases. The pseudo-code of the algorithm Movement($a_i$, $t$) is given in Algorithm \ref{Algo:1-dynamic}.
\begin{algorithm}
     \caption{Movement($a_i$, $t$)}\label{Algo:1-dynamic}
     \If{$a_i.r=1$}
    {
    It sets $a_i.r=0$.\\
    \If{attempt to move by the agent $a_i$ in round $(t-1)$ is non-successful}
    {
        $a_i.success=False$\\
    }
    \ElseIf{attempt to move by the agent $a_i$ in round $(t-1)$ is successful}
    {
        $a_i.success=True$\\
        call DFS($a_i, a_i.label$)
    }
    }   
    \If{$a_i.r=0$}
    {
    It sets $a_i.r=1$.\\    
        \If{$a_i.success=False$ }
            {
                \If{there is no other agent except $a_i$}
                    {
                    move through the port $a_i.pout$
                    }
                \ElseIf{there is another agent $a_j$ at the current node $v$}
                    {
                        \If{$a_j.success=True$}
                        {
                        move through the port $a_i.pout$
                        }
                        \ElseIf{$a_j.success=False$}
                        {
                        call Test($a_i,a_j$)
                        }
                    }
            }
        \ElseIf{$a_i.success=True$}
            {
            move through $a_i.pout$
            }
    }
 \end{algorithm}

\begin{itemize}
    \item[1.] \textbf{Case} $\mathbf{(a_i.r=1)}$: It sets $a_i.r=0$. If in the previous round $t-1$, agent $a_i$ could not move through its computed port, then it sets $a_i.success=False$, otherwise, $a_i.success=True$. Further, if $a_i.success=True$, this means $a_i$ has moved to a new node. Based on the DFS scheme, it computes the next port it has to move through and stores it in the variable $a_i.pout$. Precisely, the agent $a_i$ calls the algorithm DFS($a_i, a_i.label$). We provide the details of DFS($a_i,a_i.label$) later (lines 1-7 of Algorithm \ref{Algo:1-dynamic}).  

    \item[2.] \textbf{Case} $\mathbf{(a_i.r=0)}$: It sets $a_i.r=0$. We have the following two cases:
    \begin{itemize}
        \item[(A)] If $a_i.success=False$:  This means it has not moved to a new node. It now checks if it is alone at the current node or if any other agent, say $a_j$ is present with it. In the former case, it does not change its parameters and moves through the port $a_i.pout$ (lines 11-12 of Algorithm \ref{Algo:1-dynamic}). However, if there is another agent $a_j$ at the current node, then it decides whether to move through the computed port in the previous round, skip the current port, or restart a new DFS, based on the two cases:
        \begin{itemize}
            \item[(a)] If $a_j.success=True$: The agent $a_i$ moves through the port $a_i.pout$ (lines 14-15 of Algorithm \ref{Algo:1-dynamic}).
            \item[(b)] If $a_j.success=False$: The agent $a_i$ understands that $a_j$ also could not move in the previous round. Thus, it calls the function Test($a_i,a_j$). This function helps the agent $a_i$ to decide whether to start a new DFS traversal, skip the current port or continue to move through $a_i.pout$ based on the IDs of $a_i$ and $a_j$. The details of this function are provided later in this section (lines 16-17 of Algorithm \ref{Algo:1-dynamic}). 
        \end{itemize}
        
        \item[(B)] If $a_i.success=True$: This means the agent $a_i$ has entered into a new node, say $v$. It moves through its computed value of $a_i.pout$ (lines 18-19 of Algorithm \ref{Algo:1-dynamic}). 
    \end{itemize}      
\end{itemize}

Now we provide the details of Test($a_i,a_j$) along with the pseudo-code \ref{Algo: Test}. If more than one agent is present at a node, then the agents make decisions as follows. 
\begin{itemize}
    \item[1.] If $a_i.pout\neq a_j.pout$: The agent $a_i$ moves through $a_i.pout$ (lines 1-2 of Algorithm \ref{Algo: Test})
    \item[2.] If $a_i.pout=a_j.pout$: Then the state of both agents is checked. The possible cases are both $a_i$ and $a_j$ have state $explore$, one of them wants to $backtrack$, or both want to $backtrack$ through the same port. The decision for each case is as follows:
        \begin{itemize}
            \item[(A)] If $a_i.state=explore$ and $a_j.state=explore/backtrack$: If $a_i.ID>a_j.ID$,
            then the agent $a_i$ skips the current port i.e. $a_i.pout$, and sets $a_i.pout=(a_i.pout+1)\mod\delta$. The agent $a_i$ checks the constituents of whiteboard $wb_v(a_i.parent,\\ a_i.label)$. We have two cases. Either the current node is the root node or not. If the current node is the root node, i.e. $a_i.parent=-1$ and $a_i.pout=0$, then it increments its value of $label$ by $1$ and updates $(-1, a_i.label)$ on the whiteboard of the current node. This means $a_i$ begins a new DFS traversal of the graph with the current node as its new root node. It moves through $a_i.pout$. However, if $a_i.pout\neq 0$, then it simply moves through $a_i.pout$. On the other hand, if the current node is not the root node, then $a_i$ checks if $a_i.pout$ is the same as $a_i.parent$. If $a_i.pout=a_i.parent$, then $a_i$ updates $a_i.state=backtrack$. Otherwise, it does not update its state. It further moves through $a_i.pout$ (lines 5-19 of Algorithm \ref{Algo: Test}).
            
            If $a_i.ID<a_j.ID$, it moves through same port $a_i.pout$ (lines 18-19 of Algorithm \ref{Algo: Test}).
            \item[(B)] If $a_i.state=backtrack$ and $a_j.state=explore/backtrack$: If $a_i.ID>a_j.ID$, then $a_i$ increments its $label$ to $a_i.label+1$, sets $a_i.state=explore$ and begins a new DFS traversal of the graph with the current node as the root node and $a_i.pout$ same as the minimum port except for the current value of $a_i.pout$. The agent $a_i$ also writes ($-1,a_i.label$) on the whiteboard of the current node (lines 21-24 of Algorithm \ref{Algo: Test}). 
            However, if $a_i.ID<a_j.ID$, then it moves through its current $a_i.pout$ (lines 25-26 of Algorithm \ref{Algo: Test}).
        \end{itemize}
\end{itemize}

\begin{algorithm}
     \caption{Test($a_i$, $a_j$)} \label{Algo: Test}
     \If{$a_i.pout\neq a_j.pout$}
        {
        move through $a_i.pout$
        }
    \ElseIf{$a_i.pout=a_j.pout$}
        {
        \If{$a_i.state=explore$}
            {
            \If{$a_i.ID>a_j.ID$}
            {
                check the constituents of $wb_v(a_i.parent, a_i.label)$ \\
                set $a_i.pout=(a_i.pout+1)\mod\delta_v$\\
                \If{$a_i.parent=-1$}
                {
                    \If{$a_i.pout=0$}
                    {
                        set $a_i.label=a_i.label+1$ and $a_i.state=explore$\\
                        write $(-1,a_i.label)$ on $wb_v$ and move through $a_i.pout$
                    }
                    \Else
                    {
                        move through $a_i.pout$
                    }
                }
                \Else
                {
                    \If{$a_i.pout=a_i.parent$}
                    {
                    update $a_i.state=backtrack$\\
                    }
                move through $a_i.pout$
                }
            }
            \ElseIf{$a_i.ID<a_j.ID$}
            {
                move through $a_i.pout$
            }
            }
       
        \ElseIf{$a_i.state=backtrack$}
            {
                \If{$a_i.ID>a_j.ID$}
                {
                    $a_i$ sets $a_i.label=a_i.label+1$, $a_i.state=explore$\\
                    write $(-1, a_i.label)$ on $wb_v$\\
                    set $a_i.pout$ same as the minimum available port except for the current value of $a_i.pout$ and move through $a_i.pout$
                }
                \ElseIf{$a_i.ID<a_j.ID$}
                {
                    move through $a_i.pout$
                }
            }
        }       
 \end{algorithm}

Each agent $a_i$ runs its own DFS with label $a_i.label$ as follows. 
\begin{itemize}
    \item[1.] \textbf{Case} $\bm{(a_i.state=explore)}$: Let the agent $a_i$ be currently at node $v$. It checks the constituents of $wb_v$ corresponding to $a_i$. Let the constituents of the whiteboard be of the form $wb_v(a_i.x, a_i.y)$. If $a_i.y\neq a_i.label$ or $a_i.y=\bot$, this means the node $v$ is visited by $a_i$ for the first time with label $a_i.label$. Thus, it writes ($a_i.pin, a_i.label$) corresponding to $wb_v(a_i.x,a_i.y)$. The agent $a_i$ then updates $a_i.pout=(a_i.pin+1)\mod\delta_v$. If this updated value of $a_i.pout$ is not the same as $a_i.x$ on the whiteboard, then $a_i$ does not change its $state$. Else, if $a_i.pout=a_i.x$, then $a_i$ sets $a_i.state=backtrack$ (lines 2-7) of Algorithm \ref{Algo: DFS}.

    Otherwise, if $a_i.y=a_i.label$ on the $wb_v$, this means the node $v$ was visited beforehand by the agent $a_i$ with label $a_i.label$. Thus, it sets $a_i.state$ to $backtrack$. It also updates $a_i.pout=a_i.pin$ (lines 8-9 of Algorithm \ref{Algo: DFS}).

    \item[2.] \textbf{Case} $\bm{(a_i.state=backtrack)}$: The agent $a_i$ updates its value of $a_i.pout$ as $(a_i.pin+1)\mod\delta_v$ and checks $wb_v(a_i.x,a_i.y)$. If $a_i.x= -1$, i.e. the current node is not the root node, then it further checks if the updated value of $a_i.pout=0$ or not. If $a_i.pout=0$, then it means the graph has been explored by $a_i$ with its current label $a_i.label$. Thus, it starts a new DFS traversal of the graph by updating $a_i.state=explore$ and $a_i.label=a_i.label+1$. Further, it writes on the whiteboard $wb_v(a_i.x, a_i.y)=(-1, a_i.label)$. However, if $a_i.pout\neq 0$, then it updates $a_i.state=explore$ (lines 11-18 of Algorithm \ref{Algo: DFS}). 
    
    If the current node is not the root node, i.e. $a_i.x\neq -1$, then it checks if $a_i.pout$ is same as $a_i.x$ extracted from the whiteboard. If it is the same, then it keeps its state as $backtrack$. Otherwise, it updates its state to $explore$ (lines 19-23 of Algorithm \ref{Algo: DFS}). 
\end{itemize}
The pseudo-code is given as  DFS($a_i,a_i.label$) below. Next, we present a movement strategy involving three agents, known as the \textit{cautious walk}, which helps in determining the location of the black hole.
\begin{algorithm}
    \caption{DFS($a_i, a_i.label$)} \label{Algo: DFS}
    \If{$a_i.state=explore$}
    {
        check the constituents of $wb_v(a_i.x, a_i.y)$\\
        \If{$a_i.y\neq a_i.label$ or $a_i.y=\bot$}
        {
            write on the whiteboard $wb_v(a_i.x, a_i.y)=(a_i.pin, a_i.label)$\\
            update $a_i.pout=(a_i.pin+1)\mod\delta_v$\\
                \If{$a_i.pout=a_i.pin$ }
                {
                set $a_i.state=backtrack$\\
                }
        }
        \ElseIf{$a_i.y = a_i.label$}
        {
        set $a_i.state=backtrack$ and $a_i.pout=a_i.pin$\\
        }
    }
    \ElseIf{$a_i.state=backtrack$}
    {        
        set $a_i.pout=(a_i.pin+1)\mod\delta_v$\\
        check the constituents of $wb_v(a_i.x, a_i.y)$\\
            \If{$a_i.x=-1$}
            {
                \If{$a_i.pout=0$}
                {
                    set $a_i.state=explore$ and update $a_i.label=a_i.label+1$\\
                    write on the whiteboard $wb_v(a_i.x, a_i.y)=(-1,a_i.label)$
                }
                \ElseIf{$a_i.pout\neq 0$ }
                {
                    set $a_i.state=explore$
                }
            }
            \ElseIf{$a_i.x\neq -1$}
            {
                \If{$a_i.pout\neq a_i.x$ available in $wb_v(a_i.x, a_i.y)$}
                    {
                        set $a_i.state=explore$\\
                    }
                \ElseIf{$a_i.pout=a_i.x$}
                    {
                        set $a_i.state=backtrack$\\
                    } 
            }   
    }
 \end{algorithm}

 Let $a_1, a_2, a_3$ be three agents initially co-located at a safe node $v_r$ in the graph $G$, with unique identifiers such that $a_1.ID < a_2.ID < a_3.ID$. Let $\mathcal{P}$ denote the DFS traversal path of the static graph $G$ rooted at node $v_r$, i.e., we assume for now that all edges of $G$ are present at all times in the time-varying graph $\mathcal{G}$. We aim to establish the correctness of an exploration strategy for the three agents in the dynamic setting where, in each round, at most one edge of $G$ can be missing.

We start with the following observations of our exploration algorithm with three agents.

\begin{observation}\label{obs:path_of_exp_a1}
Let \( e = (u_1, u_2) \) be an edge that is missing from the time-varying graph \( \mathcal{G} \) at round \( r \). Let \( q_1 \) and \( q_2 \) denote the port numbers associated with edge \( e \) at nodes \( u_1 \) and \( u_2 \), respectively. Specifically, the outgoing port at \( u_1 \) that leads to \( u_2 \) is \( q_1 \), and the outgoing port at \( u_2 \) that leads to \( u_1 \) is \( q_2 \).

Suppose there are multiple agents located at node \( u_1 \) and their next move in the DFS traversal corresponds to port \( q_1 \). If any agent \( a_i \) at \( u_1 \) that is not the minimum ID agent (the smallest ID) among the agents at this node becomes aware of the missing edge \( e \), it will either initiate a new DFS traversal from the root or continue its current DFS traversal by skipping port \( q_1 \). As a result, the smallest ID agent, specifically agent \( a_1 \), does not deviate from path $\mathcal{P}$.
\end{observation}

\begin{theorem}\label{thm:1dynamic}
    One of the three agents $a_1$, $a_2$, or $a_3$ explores $G$ in $O(|E|^2)$ rounds, where each agent uses only $O(\log n)$ bits of memory.
\end{theorem}

\begin{proof}
Suppose agent $a_1$ begins executing a DFS traversal from node $v_r$. As per our algorithm, agents move in odd-numbered rounds and use even-numbered rounds to detect whether their previous movement attempt was successful. If no edge failure occurs during the first $8|E|$ rounds, i.e., all edges required by $a_1$ for its DFS traversal remain present when needed then agent $a_1$ completes its exploration of $\mathcal{G}$ by following the DFS path specified in Observation \ref{obs:path_of_exp_a1} and Theorem \ref{Thm:DFS}.\footnote{Each DFS traversal takes at most $4|E|$ moves, which correspond to $8|E|$ rounds in total.}

Now, suppose that in some odd round $r$, agent $a_1$ is at node $u$ and attempts to traverse edge $e = (u, v)$, but the edge is missing. Let $t > r$ be the next odd round in which $a_1$ retries traversing edge $e$. We show that if edge $e$ does not reappear within the next $32|E|$ rounds from round $r$, then either agent $a_2$ or $a_3$ successfully explores the entire graph.

Agent $a_2$ also performs a DFS traversal, independently initiated from its own root. Within the next $8|E|$ rounds after round $t$, agent $a_2$ must be in one of the following scenarios:

\begin{itemize}
\item \textbf{Case 1:} It is at node $v$ and attempts to traverse edge $(v, u)$, but fails, and thus becomes stuck.
\item \textbf{Case 2:} It is at node $u$ and attempts to move via edge $(u, v)$ in state $a_i.state =backtrack$.
\item \textbf{Case 3:} It is at node $u$ and attempts to move via edge $(u, v)$ in state $a_i.state =explore$.
\end{itemize}

In \textbf{Case 2}, agent $a_2$ treats node $u$ as the root and starts a new DFS traversal. Within the next $8|E|$ rounds, it either explores $G \setminus \{(u,v)\}$ or becomes stuck at node $v$ when it tries to traverse edge $(v,u)$. In \textbf{Case 3}, the agent skips edge $(u,v)$ and continues its DFS. There are two possible scenarios: either it reaches node \(v\) and remains stuck there, or it reaches the root node and initiates a new DFS. In either case, within the next \(8|E|\) rounds starting from the root, it either explores the graph \(G \setminus \{(u,v)\}\) or gets stuck at node \(v\) while attempting to traverse the edge \((v,u)\).

Hence, by round $t + 16|E|$, either agent $a_2$ has explored $G \setminus \{(u,v)\}$, or it has become stuck at node $v$. If exploration is completed by agent $a_2$, the objective is fulfilled. Otherwise, let $t' \in [t, t + 16|E|]$ be the round in which agent $a_2$ gets stuck at node $v$, and let $a_3$ be at some node $w$ at round $t'$.
Within the next $8|E|$ rounds (i.e., by round $t' + 8|E|$), agent $a_3$ falls into one of the following cases:

\begin{itemize}
\item \textbf{Case A:} It is at node $u$ (resp. $v$) and tries to move via edge $(u,v)$ (resp. $(v,u)$) in state $a_i.state = backtrack$.
\item \textbf{Case B:} It is at node $u$ (resp. $v$) and tries to move via edge $(u,v)$ (resp. $(v,u)$) in state $a_i.state = explore$.
\end{itemize}

In \textbf{Case A}, agent $a_3$ considers node $u$ (resp. $v$) as its root and initiates a new DFS. Within the next $8|E|$ rounds, it explores $G \setminus \{(u,v)\}$ because it skips the missing edge if it reaches the other endpoint. In \textbf{Case B}, the agent skips the missing edge and either:
\begin{itemize}
    \item reaches node $v$ (resp. $u$) with $a_3.state=backtrack$, reducing this to Case A; or
    \item reaches node $v$ (resp. $u$) with $a_3.state=backtrack$, again skipping the edge and continuing toward the root; or
    \item reaches the root node and initiates a new DFS, from where it again explores $G \setminus \{(u,v)\}$ within $8|E|$ rounds.
\end{itemize}

Thus, by round $t' + 16|E|$, agent $a_3$ completes exploration of $G \setminus \{(u,v)\}$. Consequently, by round $t + 32|E|$, if edge $(u,v)$ has not reappeared, then either $a_2$ or $a_3$ has successfully explored $\mathcal{G} \setminus \{e\}$, fulfilling our objective.

If the edge reappears before round $t + 32|E|$, then agent $a_1$ continues its DFS traversal along the precomputed path (Observation \ref{obs:path_of_exp_a1}), completing its exploration. Therefore, in at most $32|E| \times 8|E| = 256|E|^2$ rounds, at least one agent completes the exploration of $\mathcal{G}$. 

Each agent maintains the following variables: $a_i.\mathit{ID}$, $a_i.\mathit{success}$, $a_i.\mathit{label}$, $a_i.\mathit{state}$, $a_i.\mathit{pin}$, $a_i.\mathit{pout}$, and $a_i.r$. In each round $r$, the agent uses $O(1)$ memory to store $a_i.\mathit{success}$, $a_i.\mathit{state}$, and $a_i.r$. To store the values of $a_i.\mathit{ID}$, $a_i.\mathit{pin}$, and $a_i.\mathit{pout}$, the agent requires $O(\log n)$ memory per variable, since each represents a node identifier or a port number in an $n$ node graph. The variable $a_i.\mathit{label}$ may increase up to $O(|E|^2)$ over the course of the algorithm. However, since agent $a_i$ only need to store the current value of $a_i.\mathit{label}$, and the maximum label value is bounded by $O(|E|^2)$, it can be stored using $O(\log |E|)$ bits, which is $O(\log n)$ as $|E| \leq n^2$.
Therefore, each agent requires $O(\log n)$ memory throughout the execution of the algorithm. This completes the proof.
\end{proof}

Based on Theorem \ref{thm:1dynamic}, we have the following remark.

\begin{remark}
A permanently missing edge can indefinitely block at most two agents from completing their DFS exploration. In particular, one agent may be stuck at one endpoint of the missing edge while attempting to traverse it, and the second agent may be stuck at the opposite endpoint. All other agents, by design of the algorithm, either skip the missing edge or restart DFS, eventually enabling one of them to explore the entire graph.
\end{remark}

\subsection{Cautious walk} \label{Sec: Cautious_walk}
Cautious walk is a movement strategy that ensures that one agent stays alive if a group of agents perform this movement while exploring an edge $(u,v)$ from a node $u$ without knowing that $v$ is the black hole. 
We need 3 agents to perform a cautious walk in our model such that no more than 2 agents can be destroyed. Below, we discuss in detail.

Let $a_1$, $a_1'$, and $a_1''$ be three agents at a safe node $u$. If these agents want to traverse through the edge $(u,v)$ to reach a node $v$, it is required first to verify if the node $v$ is safe or not. One of the three agents is the leader, while the other two are known as helpers. The helpers check whether the node $v$ is safe or not so that the leader can safely visit it. Let $a_1$ be the leader and $a_1'$, $a_1''$ be the first and second helper, respectively. The first helper $a_1'$ attempts to move through edge $(u,v)$. The movement by $a_1'$ is either successful (due to the presence of the edge) or unsuccessful (due to the absence of the edge). Until the movement by $a_1'$ is successful, it re-attempts to move through edge $(u,v)$. Let at round $t$, $a_1'$ successfully move through edge $(u,v)$. Till this happens,  $a_1$, $a_1''$ waits at $u$. At round $t+1$, $a_1'$ tries to return through edge $(v,u)$ in case $v$ is a safe node so that the leader can understand that $v$ is safe. If the edge is deleted in round $t+1$, then the leader $a_1$ can not understand if the non-return of $a_1'$ is due to the fact that the edge is removed or due to the fact that $v$ is a black hole. To avoid this confusion, from $t+1$ round onwards $a_1'$ tries to return through edge $(v,u)$ in case $v$ is a safe node and $a_1''$ tries to move through edge $(u,v)$. Let $c\geq 1$ be the smallest integer such that the edge  $(u,v)$ is present in round $t+c$. Then $a_1''$ surely does a successful move in round $t+c$ through edge $(u,v)$. If $v$ is a safe node, then  $a_1'$ must have a successful return at round $t+c$ through edge $(v,u)$. In this case, the leader $a_1$ understands that the node $v$ is not a black hole. If $v$ is not a safe node, then  $a_1''$ leaves but $a_1'$ does not return, and the leader understands $v$ is a black hole; hence $a_1'$ and $a_1''$ both will be destroyed at $v$.

In case $a_1$ understands $v$ is a safe node, $a_1$ and $a_1'$ attempt to move through the edge $(u,v)$ until their movement is successful. Note that, if $v$ is a safe node, it requires 3 rounds (may not be contiguous) such that in those rounds the edge $(u,v)$ is present in the graph for the leader $a_1$ to reach $v$.
In this way, all three agents reach the safe node, or the leader finds the black hole. A flow-chart representing a cautious walk by a leader $a_1$, its first and second helpers $a_1'$ and $a_1''$ through a port $p$ is given in Figure \ref{fig:cautious_walk}.

\begin{figure*}[h!]
    \centering    
    \includegraphics[width=0.9\linewidth]{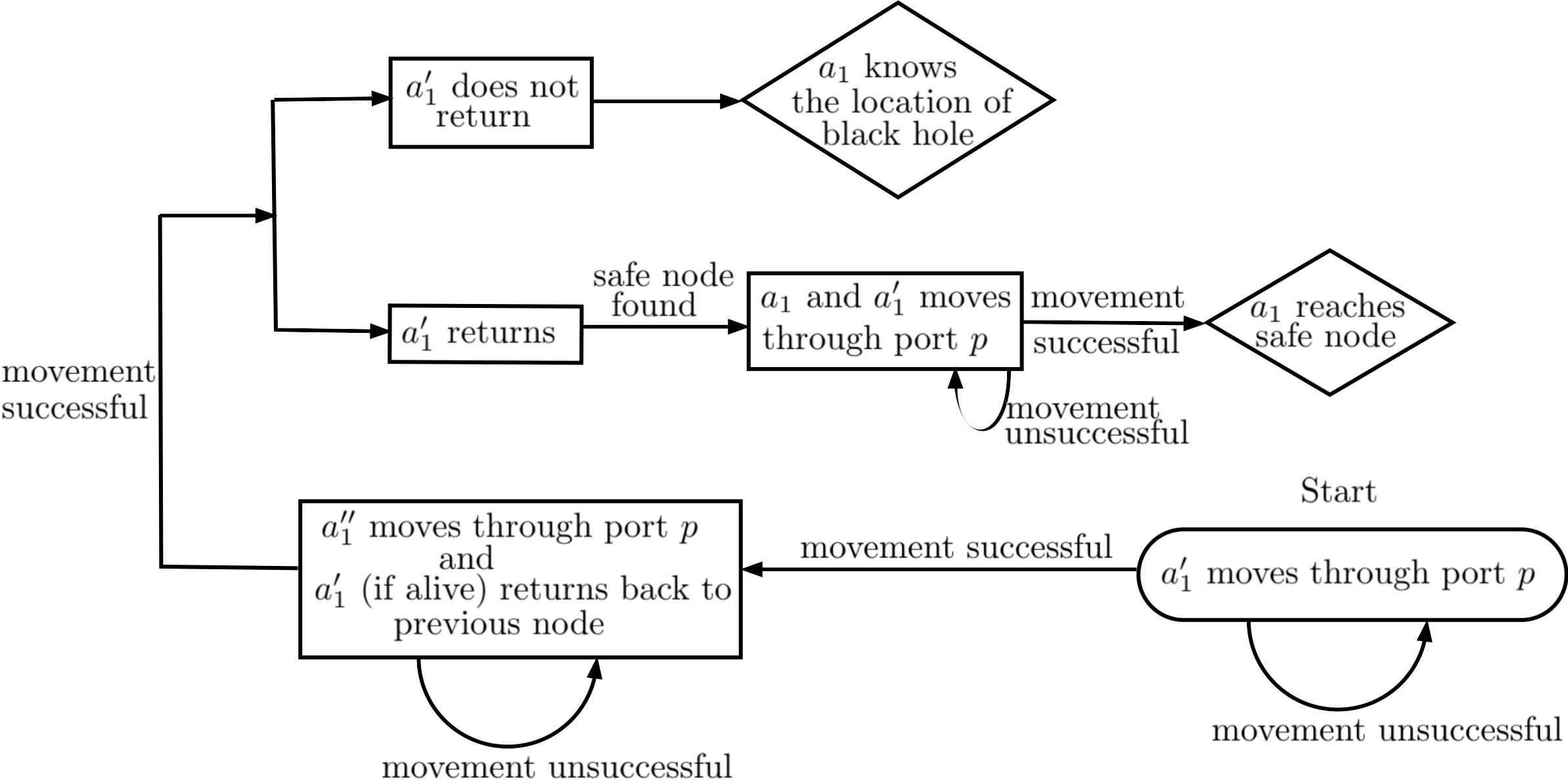}
    \caption{Flow Chart representing cautious walk by a leader ($a_1$) and helpers ($a_1'$ and $a_1''$)}
    \label{fig:cautious_walk}    
\end{figure*}


\subsection{The algorithm: multiple DFS with cautious walk}\label{sec: Algorithm} 
The algorithm to solve $1$-BHS in a general graph using $9$ agents is given in this section. 
We replace each move of an agent in the exploration technique mentioned in section \ref{Exp:1} with the cautious walk by $3$ agents as described in section \ref{Sec: Cautious_walk}. Precisely, the exploration performed by $a_1, a_2$, and $a_3$ (Section \ref{Exp:1}), is now performed by three groups that consist of three agents each. Let us suppose 9 agents are initially present at a safe node $v$ in graph $\mathcal{G}$. Also, suppose that the agents are arranged in the increasing order of their IDs, i.e. $a_1, a_2, a_3, a_1', a_2', a_3', a_1'', a_2'', a_3''$ such that $a_1$ is the agent with minimum ID and $a_3''$ with the largest ID. These agents are divided into three groups $G_1$, $G_2$, and $G_3$ as follows. The group $G_1$ consists of the agents $a_1, a_1'$, $a_1''$, where $a_1$ is the leader for $G_1$, $a_1'$ and $a_1''$ are the first and second helpers respectively. Similarly, $G_2$ consists of agents $a_2, a_2'$, $a_2''$ where $a_2$ is the leader, $a_2'$ and $a_2''$ are the first and second helpers respectively. Finally, group $G_3$ consists of agents $a_3, a_3'$, $a_3''$ where $a_3$ is the leader, $a_3'$ and $a_3''$ are the first and second helpers respectively. Each group performs cautious movement as mentioned in Section \ref{Sec: Cautious_walk}. The decision of movement based on the exploration technique is performed by the leaders $a_1, a_2$, and $a_3$. While the helpers do what their leaders instruct. Each node that consists only of a helper and is without its leader is considered to be an empty node. Several situations may arise due to a missing edge where these helpers may need to change the groups that were formed initially. We describe all such scenarios below.

Let us suppose $G_i$, which consists of $a_i$, $a_i'$, and $a_i''$, is currently at a node $u$. Suppose, according to the exploration technique, $G_i$ has to move through port $p$ to reach node $v$. It thus has to perform cautious movement. During the cautious movement, the first helper $a_i'$ moves through port $p$ at round $t$. This is followed by the return of the first helper back to $u$ and the movement of the second helper $a_i''$ through port $p$ at some later round when the edge is present. Finally, if the first helper returns to $u$, the leader understands that $v$ is safe. Thus, the leader, along with the first and second helpers, moves through port $p$ at some round. However, three possible scenarios may arise due to the missing edge:
\begin{itemize}
     \item[1.] The first helper $a_i'$ could not move through port $p$ at some round $t$ due to missing edge $(u,v)$. In this case, the agents $a_i$, $a_i'$, and $a_i''$ are present at the node $u$. If another group, say $G_j$ (i.e. the group that contains $a_j$), meets with the agents of $G_i$. If $i<j$, then it does not deviate from its path and performs its cautious walk to reach node $v$. If $i>j$, then $G_i$ has to deviate (leader $a_i$ in particular), then the leader $a_i$ decides based on the exploration technique and the helpers $a_i'$ and $a_i''$ move with their leader through cautious movement. No group change is required here.
    \item[2.] The first helper $a_i'$ moves through port $p$ at some round, and at the beginning of the next round, edge $(u,v)$ is missing. The leader $a_i$ and the second helper $a_i''$ are present at node $u$. Suppose group $G_j$ reaches node $u$. If $j>i$, then no group change is required here.
    
    If $j<i$, then based on the exploration technique, $G_i$ (leader $a_i$ in particular) has to move through a new port or start a new DFS. Then, it does not wait for its first helper $a_i'$ to return, instead, the first helper $a_j'$ of $G_j$ joins the group $G_i$ and becomes the first helper in group $G_i$. Since $G_j$ has to move through port $p$ (that is why $G_i$ has to skip this port or start a new DFS), it considers that its first helper is stuck at the other end of the missing edge. And agent $a_i'$ is included in $G_1$ when the edge appears as the first helper.    
    \item[3.] The second helper $a_i''$ moves through port $p$ while the first helper $a_i'$ returns to $u$ at some round and then at the beginning of the next round, the edge $(u,v)$ is deleted by the adversary: In this case, $a_i$ and $a_i'$ are present at node $u$. Similar to case 2, suppose group $G_j$ reaches node $u$. If $j>i$, then no group change is required here.

    If $j<i$, then $G_i$ (the leader $a_i$ in particular) has to skip the current port or start a new DFS; it does not wait for $a_i''$. Whereas it includes the second helper from group $G_1$, i.e. $a_j''$ in its group and continues its exploration. The leader of $G_j$ considers that its second helper is stuck at the other end of the missing edge. And, agent $a_i''$ is included in $G_j$ when the edge appears as the second helper.
\end{itemize}
Note that the decisions based on the exploration technique are made only when one leader meets with another leader. Otherwise, if a leader visits a node where a helper from another group is present, it considers the visited node as an empty node and proceeds further based on the exploration technique. In the following section, we show that one of the agents survives and can identify which port leads to the black hole.

\subsection{Correctness and analysis}

\begin{theorem}\label{thm:exp_grp}
    At least one of $G_1$, $G_2$, or $G_3$ explores the graph $G$ correctly in $O(|E|^2)$ rounds.
\end{theorem}
\begin{proof}
 The exploration strategy performed by agents $a_1$, $a_2$, and $a_3$ (given in the Section \ref{Exp:1}) is replicated by three groups $G_1$, $G_2$, and $G_3$ consisting of three agents each (given in the Section \ref{sec: Algorithm}). As per the strategy discussed in the Section \ref{sec: Algorithm}, there are three groups, $G_1$, $G_2$, $G_3$, where $G_1$ consists of agents $a_1$, $a_1'$, and $a_1''$, $G_2$ consists of agents $a_2$, $a_2'$, and $a_2''$, and $G_3$ consists of agents $a_3$, $a_3'$, and $a_3''$ in the initial configuration, i.e., before any change of groups is done. The group $G_i$ replicates the movement strategy of exploration of agent $a_i$ for $i=\{1,2,3\}$. However, during this replication, several cases arise where a group may need to skip an edge or start a new DFS traversal (as per the exploration strategy using three agents), but at least one of its helpers is separated from the leader of the group. In this scenario, it is important to ensure that the decision taken by a group $G_i$ is the same as the decision taken by the agent $a_i$ in the exploration strategy. The following cases are possible: either $G_1$ and $G_2$ meets, or $G_1$ and $G_3$ meets, or $G_2$ and $G_3$ meets. However, during this meeting, it may also occur that all the three agents of the group are not together. Without loss of generality, let the leader of $G_2$ (i.e., $a_2$) be present at a node $u$ and it is stuck here due to the presence of a missing edge $(u,v)$. The following cases arise:

    \noindent\textbf{Case 1} All the agents of $G_2$ are colocated at $u$: In this case, the agents of $G_2$ are stuck at $u$ due to the missing edge $(u,v)$. It may happen that $G_1$ or $G_3$ may visit $u$ and have to take the same edge $(u,v)$. If $G_1$ visits, then as per the exploration strategy, $G_1$ must wait at $u$ and $G_2$ must skip this edge (if $G_2$ is in $explore$ state) or start a new DFS from $u$ (if $G_2$ is in $backtrack$ state). Since all the agents of $G_2$ are together, it can proceed as per the decision of the exploration strategy without any need for a change of groups. Thus, in this case, the traversal path of $G_2$ is exactly the same as that of $a_2$ (as per the exploration strategy).

    \noindent\textbf{Case 2} The second helper of $G_2$ is present at $u$ along with $a_2$: This implies the first helper of $G_2$ is present at $v$. In this case, $G_2$ remains stuck at $u$ until another group visits $u$. Either $G_1$ or $G_3$ (with all three agents co-located) visits $u$.
    \begin{itemize}
      \item   If $G_1$ visits $u$, then $G_1$ must remain at its original path and $G_2$ must deviate. Thus, in order to execute this, a change of group occurs. $G_2$ takes the first helper of $G_1$ and makes its own first helper. Hence, $G_2$ now has all three of its agents, so it can continue its traversal. $G_1$ (with $a_1$ and third helper) remains at $u$ and understands that the (old) first helper of $G_2$ is now the (new) first helper of $G_2$. This is the same as when $a_2$ is stuck at a node and $a_1$ visits that node attempting to move through the same edge, then $a_1$ gets stuck and $a_2$ continues its traversal either by skipping the edge or starting a new DFS traversal. 
       \item If $G_3$ visits $u$, then no change of groups is required in that case. As all the agents of $G_3$ are co-located, they can continue their traversal. 
    \end{itemize}
Thus, the movement strategy of the groups is exactly the same as the movement strategy of the agents in the exploration by 3 agents.

    \noindent\textbf{Case 3} The first helper of $G_2$ is present at $u$ along with $a_2$: This case is analogous to case 2. Here, the change of groups happens for the second helper of $G_2$. Thus, the traversal path of groups is the same as that of the corresponding agents in the exploration strategy.

   By virtue of Theorem \ref{thm:1dynamic} and since the groups $G_1$, $G_2$, and $G_3$ replicates the traversal strategies of the agents $a_1$, $a_2$, and $a_3$ in the exploration strategy, at least one of the groups explores the graph $\mathcal{G}$ correctly. Each edge movement by an agent $a_i$ (in the exploration strategy) is replicated by the group $G_i$ (in our algorithm) in $O(1)$ rounds. Thus, as per the theorem \ref{thm:1dynamic}, the time taken by at least one group $G_i$ to explore the graph $\mathcal{G}$ is $O(|E|^2)$ rounds.
\end{proof}

\begin{theorem}\label{thm:main1BHS}
The 1-BHS problem can be solved in $O(|E|^2)$ rounds by a team of $9$ agents, each equipped with $O(\log n)$ bits of memory, starting from a rooted initial configuration. Furthermore, each node requires only $O(\log n)$ bits of storage.
\end{theorem}

\begin{proof}
   As per Theorem \ref{thm:exp_grp}, each node is visited by at least one of the groups $G_1$, $G_2$, or $G_3$ in $O(|E|^2)$ rounds. Since each group proceeds according to the cautious walk, and a black hole can be determined by a group of three agents using the cautious walk strategy, it is ensured that at least one of the groups determines the black hole. Hence, 1-BHS can be solved by a team of 9 agents in $O(|E|^2)$ rounds. 
    
    In Theorem \ref{thm:1dynamic}, we have shown that each agent utilizes \( O(\log n) \) memory to maintain all parameters. In our algorithm for solving 1-BHS, agents need to remember who the helpers, leaders, and members of the groups are. This can also be accomplished using \( O(\log n) \) memory. Therefore, agents need $O(\log n)$ memory to solve 1-BHS.

    At each node $v$, agent $a_i$ writes $wb_{v}(a_i.\mathit{parent}, a_i.\mathit{label})$. Since the number of agents is constant, and $a_i.\mathit{parent}$ represents a port number, it can be stored using $O(\log n)$ bits. The value $a_i.\mathit{label}$ is bounded by $O(|E|^2)$, and since $|E| \leq n^2$, this implies $a_i.\mathit{label}$ can be stored using $O(\log |E|) = O(\log n)$ bits. Hence, the information written at each node requires $O(\log n)$ bits. This completes the proof. 
    \end{proof}

\section{Impossibility of $f$-BHS}\label{sec:imp_fbhs}
In this section, we show that it is impossible to find the black hole in a time-varying graph $\mathcal{G}$ using $2f+1$ agents given that the adversary can remove at most $f$ edges from $G$ (i.e., footprint graph). Our impossibility result is as follows.

\begin{theorem}
Let $\mathcal{G}$ be a TVG in which each node has infinite storage. Then, the $f$-BHS problem cannot be solved by $2f+1$ mobile agents, even if the agents have infinite memory, are initially co-located at a safe node, and have no prior knowledge of the value of $f$.
\end{theorem}

\begin{proof}
 We prove the above statement by contradiction. Let us consider a clique $K_{f+2}$ of size $f+2$, and one of the nodes of $K_{f+2}$ is a black hole. Let $v_0$, $v_1$, $v_2$, $v_3$, \ldots, $v_{f+1}$ be the vertices of $K_{f+2}$ and edge $e_i$ connects two nodes $v_i$ and $v_{0}$, where $i \in [1,\, f+1]$. In $K_{f+2}$, $e_{ij}$ denotes the edge between nodes $v_i$ and $v_j$, where $i,\, j \in [1,\, f+1]$ and $i < j$. Without loss of generality, let $v_0$ be a black hole. An example corresponding to $f=3$  is shown in Fig. \ref{fig:k-dyn_lower_bound}.

 Let the agents $\{a_1, a_2, ..., a_{2f+1}\}$ be co-located at safe node of $K_{f+2}$. Without loss of generality, let all agents be at node $v_1$. Suppose there exists an algorithm $\mathcal{A}$, which successfully finds the black hole position in $\mathcal{G}$ with $2f+1$ agents. Let the adversary follow the following rules to remove the edges from the graph at the beginning of round $t$ when the algorithm $\mathcal{A}$ proceeds: 

\begin{itemize}
\item\textbf{R1:} An agent $a$ is at node $v_i$, with no information regarding another agent already traveled through edge $e_i$. If agent $a$ moves through edge $e_i$ in round $t$, then $e_i$ is not removed.
    \item \textbf{R2:} An agent $a$ is present at node $v_i$, and it understands that some other agent already traveled through edge $e_i$. If agent $a$ moves through edge $e_i$ in round $t$, then $e_i$ is removed.
    \item \textbf{R3:} If two or more agents are at node $v_i$ and some of the agents or all the agents want to go through edge $e_i$, then edge $e_i$ is removed.
     \item \textbf{R4:} If two or more agents are at $v_i$ and some or all of them want to go through $e_{ij}$, then edge $e_{ij}$ is not removed.
   \item \textbf{R5:} In round $t$, suppose there are at most $f$ agents in $K_{f+2}$, and each agent is at some safe nodes of $K_{f+2}$. These agents can move through at most $f$ edges at round $t$, and the adversary can remove at most $f$ edges as well. Therefore, the adversary stops their movement by removing those edges.
\end{itemize}
    
     
  
Initially, $2f+1$ agents are co-located at node $v_1$. All agents are executing the algorithm $\mathcal{A}$. At some round $t$, if the agents do not move through any edge $e_i$ at all, then it is never possible to determine the black hole. Therefore, at least one agent needs to go through an edge $e_i$ from node $v_i$, where $i\geq 1$. According to the rules mentioned above, it is possible only when rule R1 is true, i.e., agent $a$ tries to move from a node $v_i$ towards the node $v_{0}$, and it has no prior information that some agent moves through edge $e_i$ in an earlier round. Suppose at round $t$, one agent dies at $v_{0}$ when it moves through edge $e_i$, and all remaining agents have the information, i.e., the port information corresponding to edge $e_i$ used by the agent to move. At the round $t+1$, we have $2f$ many agents in graph $\mathcal{G}$ and $f$ nodes without any information on the whiteboard. At the round $t'>t$, a team of agents or a single agent reaches node $v_i$. Whenever an agent tries to go through edge $e_i$ (according to the information it has), the adversary removes the edge $e_i$ at the beginning of round $t'$ to block its movement (using rule R2). In such a way, at least $f+1$ agents die at node $v_0$, and at most $f$ many agents have information, i.e., through which $e_i$ agents have moved through from node $v_i$ in earlier rounds. In the system, at most $f$ agents are left and they are at safe nodes of $K_{f+2}$. As the adversary has the capability to remove at most $f$ edges in each round, with no more than $f$ agents remaining in the system, the movement of agents can be permanently blocked (using rule R5). This completes the proof.
    \end{proof}

\begin{figure}
    \centering
    \includegraphics[width=0.4\linewidth]{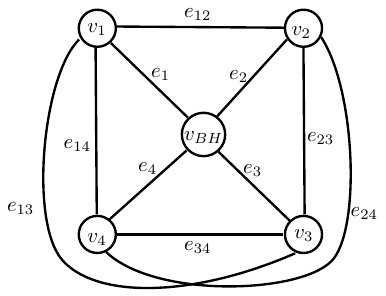}
    \caption{Construction of the graph $K_5$ when $f=3$ for impossibility of $f-$BHS}\label{fig:k-dyn_lower_bound}
\end{figure}

\section{Algorithm for $f$-BHS}\label{sec:f dynamic}
In this section, we'll talk about how to find the black hole with the help of $6f$ agents starting from a rooted initial configuration in $G$ when the adversary can remove at most $f$ edges from $G$ at each round. Again our approach is to do a graph exploration by cautious walk. Here we use the exploration strategy of \cite{Nicolo_21}. In Section \ref{sec:pre}, we discuss one existing algorithm for the exploration $\mathcal{EXPD}$ for $l$-bounded 1-connected graph $\mathcal{G}$. Recall from Section \ref{sec:pre}, $l$-bounded 1-connected $\mathcal{G}$ means that in each round $r$, $\mathcal{G}_r$ is connected and the adversary can remove at most $l$ edges from $G$ (footprint of $\mathcal{G}_r$). In our model, we consider that $\mathcal{G}$ is a time-varying graph (TVG) in which at most $f$ edges can be deleted by the adversary from $G$ in each round. It is not difficult to observe that $0\leq l \,(=f)$. The TVG $\mathcal{G}$ is connected at each round. Therefore, we can say that $\mathcal{G}$ is $f$-bounded 1-connected graph. According to Theorem \ref{thm:EXP_algo} in Section \ref{sec:pre} from \cite{Nicolo_21}, $2f$ many agents can explore the graph if there is a leader. Their algorithm requires $O(\log n)$ storage per node in the form of whiteboard. In our case, since all agents are at the same node initially, they can fix the minimum ID as the leader ID, and the remaining IDs are non-leader IDs. As we are on a $f$-bounded 1-connected graph, we use the below result (refer Theorem \ref{thm:EXP_algo} Section \ref{sec:pre}) for exploring a $f$-bounded 1-connected TVG without any black hole.

\begin{theorem} \label{thm:EXP_algo1}
A graph $\mathcal{G}$ in which the adversary can delete at most $f$ edges from $G$ in each round, $ 2f$ many agents can explore such a time-varying graph within $\Delta^{n} (\Delta +1)^{2f+n} (n-1)^{2f}$ rounds, where $\Delta$ is the maximum degree of $G$, provided each node is equipped with $O(\log n)$ storage in the form of whiteboard.
\end{theorem}
\begin{proof}
    In \cite{Nicolo_21}, the algorithm $\mathcal{EXPD}$ uses a leader to solve the exploration problem for $l$-bounded 1 connected graph $\mathcal{G}$. In our model, there is no leader. All agents are at the same node initially; they can fix the minimum ID as the leader ID, and the remaining IDs are non-leader IDs. They execute the algorithm $\mathcal{EXPD}$. Since the algorithm guarantees the exploration for $l$-bounded 1-connected graph. Due to $0\leq l \,(= f)$, $\mathcal{G}$ is $f$-bounded 1-connected graph. Therefore, $\mathcal{EXPD}$ guarantees the exploration in TVG $\mathcal{G}$ with the help of $2f$ many agents when the adversary can remove at most $f$ edges from $G$ at each round. The time complexity for exploration is the time complexity of $\mathcal{EXPD}$, i.e., $\Delta^{n} (\Delta +1)^{2f+n} (n-1)^{2f}$ rounds, where $\Delta$ is the largest degree of $G$.
\end{proof}

\noindent\textbf{The algorithm:} To solve the $f$-BHS problem, we consider $6f$ many agents at a safe node of $G$. Let $a_1$, $a_2$, \ldots, $a_{6f}$ agents are at some node $v$. Without loss of generality, let $a_i.ID<a_{i+1}.ID$. We create $2f$ many groups of agents, say $G_1$, $G_2$, \ldots, $G_{2f}$. Let $a_{3i-2}$, $a_{3i-1}$ and $a_{3i}$ are in group $G_i$, where $i \in [1, \;2f]$. Initially, all agents consider $G_1$ as the leader group and the remaining groups as non-leader groups. Each group follows the algorithm $\mathcal{EXPD}$ as follows. For each round of $\mathcal{EXPD}$, each group $G_i$ performs cautious walk as mentioned in Section \ref{Sec: Cautious_walk}. Theorem $\ref{thm:EXP_algo1}$ guarantees the exploration of $G$. This means each node is visited by at least one group $G_i$. And, cautious walk (Section \ref{Sec: Cautious_walk}) ensures at least one agent learns the black hole. Note that each of the $2f$ groups have constant many agents. According to Theorem \ref{thm:EXP_algo1}, their algorithm requires $O(\log n)$ storage at each node. Thus, our algorithm also requires $O(\log n)$ storage at each node.

\begin{theorem} \label{thm:BHS_algo}
The problem of $f$-BHS can be solved by $6f$ agents within $3\Delta^{n} (\Delta +1)^{2f+n} (n-1)^{2f}$ rounds, where $\Delta$ is the maximum degree of $G$. 
\end{theorem}

\begin{observation}
    If $f=1$ in this section, our algorithm can solve 1-BHS with the help of six agents, but it may take exponential time to {find} the black hole. Recall that, the algorithm in Section \ref{sec:1-dynamic} can solve 1-BHS with the help of 9 agents, but it takes $O(|E|^2)$ rounds. Therefore, despite requiring more agents, the algorithm in Section \ref{sec:1-dynamic} is faster than the one in this section.  
\end{observation}

\section{$1$-BHS from Rooted Configuration using 6 Agents}\label{sec:algo_6agents}
In this section, we improve the exploration strategy presented in Section \ref{Exp:1} by introducing an exploration algorithm using two mobile agents.\footnote{A similar strategy has been used in \cite{saxena_2024} to achieve dispersion.} This improved strategy ensures the complete exploration of a dynamic graph $\mathcal{G}$ under the assumption that it contains no black hole. Eventually, this helps to solve the 1-BHS problem using 6 agents starting from a rooted configuration in the footprint $G$.

\subsection{Exploration strategy with correctness}\label{EXP:2}

Let $A_1$ and $A_2$ be the two agents initially positioned at a node, say $v$, and $A_1.ID<A_2.ID$. Let the graph $G$ has no black hole. Our objective is that each node of $G$ is visited by at least one of $A_1$ and $A_2$. The exploration strategy proceeds as follows. Both agents begin their DFS traversal of the graph and continue as long as they do not encounter a missing edge. Suppose an agent reaches a node $u$ and needs to traverse port $p_i$ to reach the adjacent node $v$. If the edge $(u, v)$ is missing, then movement through port $p_i$ is not possible in that round. In this case, agent $A_1$ remains at $u$ and attempts to traverse port $p_i$ in subsequent rounds until it succeeds. Agent $A_2$, on the other hand, skips port $p_i$ and continues its DFS traversal or restarts its DFS. It is important to note that $A_1$ never deviates from its original DFS path. The agent $A_2$ continues skipping the missing ports if it is in $explore$ state. It continues doing so unless it has either completed its current DFS traversal and has no further ports left to explore or it has to backtrack via a missing port. In either situation, $A_2$ starts a new DFS traversal from its current node. Figure \ref{fig:exploration} describes the idea of our exploration strategy and Figure \ref{fig:exploration1} implements the idea on a graph. Similar to Section \ref{sec:1-dynamic}, agents execute their DFS in odd rounds, while even rounds are used to understand the success of movements from the previous odd round, and all parameters of agents are the same as mentioned in Section \ref{sec:1-dynamic}. 

\begin{figure}
    \centering
    \includegraphics[width=0.2\linewidth]{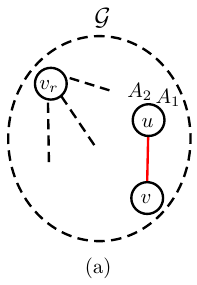}
    \hspace{0.8cm}
    \includegraphics[width=0.2\linewidth]{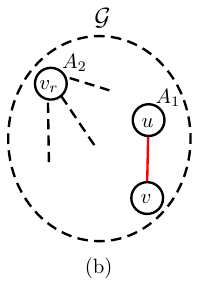}
    \hspace{0.8cm}
    \includegraphics[width=0.2\linewidth]{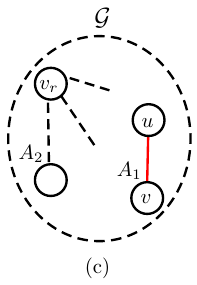}
    \caption{(a) Let TVG $\mathcal{G}$ has no black hole and agents $A_1$ and $A_2$ be initially positioned at a node $v_r$ (root node). Let the agents reach $u$, and they encounter a missing edge, (b) $A_1$ waits at $u$ but $A_2$ skips the edge $(u,v)$ and reaches at $v_r$ within $4|E|$ rounds if the edge $(u,v)$ is missing continuously for $4|E|$ rounds. Within another (at most) $4|E|$ rounds $A_2$ explores $G\setminus(u,v)$ if edge $(u,v)$ is missing, (c) If the edge $(u,v)$ appears before $A_2$ could explore the graph $G\setminus(u,v)$, $A_1$ moves a step forward.}
    \label{fig:exploration}
\end{figure}
\vspace{-0.2cm}
\begin{figure}
    \centering
    \includegraphics[width=0.25\linewidth]{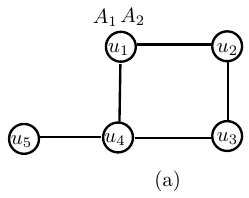}
    \hspace{0.8cm}
    \includegraphics[width=0.25\linewidth]{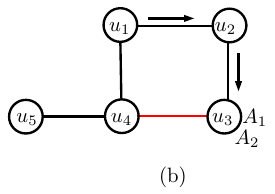}\\
\includegraphics[width=0.25\linewidth]{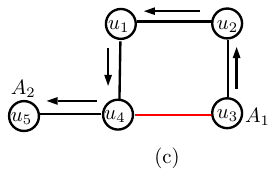}
\hspace{0.8cm}
    \includegraphics[width=0.25\linewidth]{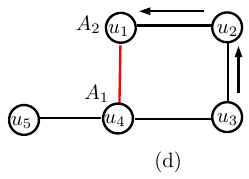}
    \caption{(a) Let $A_1$ and $A_2$ be present at $u_1$ in the initial configuration. (b) Both agents move as $u_1 \rightarrow u_2\rightarrow u_3$ and then encounter a missing edge. (c) $A_1$ waits at $u_3$ for the missing edge to reappear and $A_2$ moves as $u_3 \rightarrow u_2\rightarrow u_1 \rightarrow u_4 $. It attempts to move via edge $(u_4,u_3)$ but finds it missing, so $A_2$ skips it. Thus, it moves to reach $u_5$ and the exploration of the graph is done when the edge $(u_3,u_4)$ remains missing. (d) If $A_2$ finds a missing edge at $(u_1,u_4)$, then $A_1$ proceeds forward and reaches $u_4$. Since the movement by $A_2$ is continued even when it encounters a missing edge, this guarantees exploration of the graph.}
    \label{fig:exploration1}
\end{figure}

Now, we proceed with the correctness of our exploration strategy. Assume that $\mathcal{P}$ denotes the DFS traversal path of the static graph $G$ rooted at a node $v_r$, i.e., when all edges of $G$ are present at all times in the time-varying graph $\mathcal{G}$. We aim to establish the correctness of an exploration strategy for the two agents in the dynamic setting where, in each round, at most one edge of $G$ can be missing.

We start with the following observation of our exploration algorithm with two agents. 

\begin{observation}\label{lem:A1_explore}
    If the adversary removes an edge $e$, agent $A_1$ remains at its current node and waits for the missing edge to reappear. In contrast, agent $A_2$ either skips the edge $e$ and continues its DFS traversal or initiates a new DFS from the current node. Notably, agent $A_2$ is never blocked by a missing edge, while agent $A_1$ strictly follows path $\mathcal{P}$ without deviation. As a result, any missing edge can obstruct the progress of at most one agent.
\end{observation}


 \begin{theorem}\label{thm:expwith2}
    At least one agent explores $\mathcal{G}$ in $O(|E|^2)$ rounds, where each agent uses only $O(\log n)$ bits of memory.
\end{theorem}
\begin{proof}
Suppose $A_1$ begins executing a DFS traversal from node $u_1$. As per the already discussed mechanism, agents move in odd numbered round and use even numbered rounds to detect whether their previous attempt to move was successful or not. If no edge-failure occurs within the first $8|E|$ rounds, then $A_1$ successfully explores the entire graph $G$.

Now suppose, at some round $r$, agent $A_1$ is at a node $u$ and it has to move through edge $(u,v)$ which is missing. Let $t>r$ be the next odd round in which $A_1$ retries to move through the edge $(u,v)$. We show that if this edge remains missing for another $16|E|$ rounds from round $r$, then $A_2$ successfully explores the graph. Agent $A_2$ also performs a DFS traversal independently initiated from the root node. Within the next $8|E|$ rounds after $t$, $A_2$ is possible to be in either of the following cases:
\begin{itemize}
    \item \textbf{Case 1}: It reaches the node $u$ (or $v$) in the $explore$ state and has to move through the edge $(u,v)$ (or $(v,u)$) but is unsuccessful in doing so.
    \item \textbf{Case 2}: It reaches the node $u$ (or $v$) in the $backtrack$ state and has to move through the edge $(u,v)$ (or $(v,u)$) but is unsuccessful in doing so.
\end{itemize}

In \textbf{Case 1}, $A_2$ skips the edge $(u,v)$ (or $(v,u)$) and continues its traversal further. Thus, within the next $8|E|$ rounds, $A_2$ either explores the graph $G\setminus(u,v)$ or reaches at the root node. After reaching at the root node, $A_2$ re-starts a new DFS traversal. Hence, by $t+16|E|$ rounds, $A_2$ explores the entire graph $G\setminus (u,v)$ as it never waits for the missing edge to reappear. In \textbf{Case 2}, $A_2$ starts a new DFS traversal from $u$ (or $v$). Thus, within the next $8|E|$ rounds $A_2$ either explores $G\setminus (u,v)$ or $A_2$ reaches at the root node. Hence, by $t+16|E|$ rounds, $A_2$ explores the entire graph $G\setminus(u,v)$ as it never waits for the missing edge to reappear.

Now, if the edge $(u,v)$ is present at a round $t'$ such that $t'<t+16|E|$, then $A_1$ moves at least one edge forward. Thus, in at most $16|E|\times8|E|=128|E|^2$ rounds, at least one agent completes the exploration of $\mathcal{G}$.

Each agent $A_i$ maintains the following variables: $A_i.\mathit{ID}$, $A_i.\mathit{success}$, $A_i.\mathit{label}$, $A_i.\mathit{state}$, $A_i.\mathit{pin}$, $a_i.\mathit{pout}$, and $A_i.r$. In each round $r$, the agent uses $O(1)$ memory to store $A_i.\mathit{success}$, $A_i.\mathit{state}$, and $A_i.r$. To store the values of $A_i.\mathit{ID}$, $A_i.\mathit{pin}$, and $A_i.\mathit{pout}$, the agent requires $O(\log n)$ memory per variable, since each represents a node identifier or a port number in an $n$ node graph. The variable $A_i.\mathit{label}$ may increase up to $O(|E|^2)$ over the course of the algorithm. However, since agent $A_i$ only need to store the current value of $A_i.\mathit{label}$, and the maximum label value is bounded by $O(|E|^2)$, it can be stored using $O(\log |E|)$ bits, which is $O(\log n)$ as $|E| \leq n^2$.
Therefore, each agent requires $O(\log n)$ memory throughout the execution of the algorithm. This completes the proof. 
\end{proof}

\begin{lemma}\label{lm:storage_1}
    Our exploration requires storage of $O(\log n)$ at each node.
\end{lemma}
\begin{proof}
   At each node $v$, agent $a_i$ writes $wb_{v}(a_i.\mathit{parent}, a_i.\mathit{label})$. Since the number of agents is constant, and $a_i.\mathit{parent}$ represents a port number, it can be stored using $O(\log n)$ bits. The value $a_i.\mathit{label}$ is bounded by $O(|E|^2)$, and since $|E| \leq n^2$, this implies $a_i.\mathit{label}$ can be stored using $O(\log |E|) = O(\log n)$ bits. Hence, the information written at each node requires $O(\log n)$ bits. This completes the proof.
\end{proof}

\subsection{The algorithm: multiple DFS with cautious walk}
Since our graph $G$ contains a black hole, following only the exploration strategy described above may lead to the agents being destroyed. Therefore, the agents must proceed with caution. During the execution of our algorithm, the agents employ the cautious walk technique introduced in Section \ref{Sec: Cautious_walk}. Hence, the exploration strategy described earlier is integrated with the cautious walk strategy to address the $1$-BHS problem. We employ a total of $6$ agents to solve this problem. We now present our algorithm.

\noindent\textbf{Description:} Let $\{a_1, a_2, a_1', a_1'', a_2', a_2''\}$ be the six mobile agents, all initially co-located at a single node in the graph. Since our exploration strategy requires $2$ agents, these six agents are partitioned into two groups, $G_1$ and $G_2$, each consisting of three agents. Specifically, the role of $A_1$ (in the exploration strategy) is carried out by the group $G_1$, and the role of $A_2$ is carried out by the group $G_2$. Initially, $G_1$ consists of agents $a_1$, $a_1'$, and $a_1''$, where $a_1$ functions as the leader, and $a_1'$ and $a_1''$ act as the first and second helpers (as per the cautious walk), respectively. Similarly, $G_2$ initially comprises of agents $a_2$, $a_2'$, and $a_2''$, with $a_2$ designated as the leader and $a_2'$, $a_2''$ as the first and second helpers, respectively. It is important to note that, due to the presence of missing edges and depending on the specific execution scenario, the roles of leader and helpers within each group may dynamically change during the course of the algorithm. 

All agents begin from a single node, say $v_r$. Both groups, $G_1$ and $G_2$, initiate a Depth-First Search (DFS) traversal from this node, proceeding with their exploration unless a missing edge is encountered. As the traversal progresses, each group writes its traversal information on the whiteboard associated with every visited node. Note that $G_2$ may initiate a new DFS traversal at a later point during the algorithm’s execution; therefore, it maintains an additional parameter to keep track of this. Formally, the parameters maintained by $G_1$ and $G_2$ on the whiteboard of a node, say $v$, are as follows: 
 \begin{itemize}
    \item $wb_v(G_1.ID, parent)$- This parameter stores the information corresponding to the group $G_1$. It stores the ID of the leader of $G_1$, along with the value of $parent$ port of node $v$ according to the DFS traversal run by $G_1$.
    \item $wb_v(G_2.ID, parent, dfs\_label)$- This parameter stores the information corresponding to the group $G_2$. It stores the ID of the leader of $G_2$, along with the value of $parent$ port of node $v$ according to the DFS traversal run by $G_2$. The DFS number that is being run by $G_2$ is stored in $label$.
 \end{itemize}

 Note that $G_1$ performs exactly one DFS traversal. Thus, it does not require a counter for the number of DFS traversal it has carried out. While, $G_2$ may run multiple DFS traversals, and therefore it maintains a counter for it. This helps to distinguish the information of the previous DFS by $G_2$ from the current DFS by $G_2$. Both the groups $G_1$ and $G_2$ start the DFS traversal by replacing each movement with a cautious walk strategy (as given in Section \ref{Sec: Cautious_walk}) and continue their traversal together unless they encounter a missing edge. On encountering a missing edge, $G_1$ does not deviate (as per the exploration strategy of $A_1$). The group $G_2$ can skip the edge or restart a new DFS traversal when it encounters a missing edge (as per the exploration strategy of $A_2$). However, since the groups $G_1$ and $G_2$ are performing cautious movement, the group $G_2$ can continue its traversal by skipping the missing edge or restarting a new DFS traversal only when all three agents of $G_2$ are present. It may happen that while performing cautious walk by $G_2$, the agents (that belonged to $G_2$) are divided and are stuck on the opposite sides of the missing edge. In this scenario, $G_2$ can not continue its traversal unless all three agents (that are a part of $G_2$) are together. Thus, a change of groups occurs in this case. We provide all the possible cases and their resolution below. Note that this change of groups can never be required when $G_2$ is present at a node and it wants to backtrack via an edge, say $e$, but $e$ is missing. This is due to the fact that $G_2$ does not need to do cautious movement while backtracking from a node. For each of the following cases, we include a corresponding figure, as the group change mechanism is comparatively more complex.
 
 We resolve the following cases for $G_2$ in case it skips an edge in its current DFS, or restart a new DFS while also performing cautious movement. The cases are as follows:
 \begin{itemize}
     \item The first helper $a_2'$ could not move through port $p$ at some round $t$ due to missing edge $(u,v)$. This implies, all the agents of $G_2$ ($a_2$, $a_2'$, and $a_2''$) are present at the node $u$. Thus, $G_2$ can continue its traversal (skip the edge or restart a new DFS traversal) and no change of groups is required in this case. Figure \ref{fig:nochange} describes this case where no change in group configuration is required. Note that in our figures we use $l^{G_1}$, $h_1^{G_1}$, and $h_2^{G_1}$ to denote the leader, first helper, and second helpers of $G_1$ respectively. Analogously, we use $l^{G_2}$, $h_1^{G_2}$, and $h_2^{G_2}$ to denote the leader, first helper, and second helpers of $G_2$.
     
     \item The first helper of $G_2$ ($a_2'$) moves through port $p$ at some round, and at the beginning of the next round, edge $(u,v)$ is missing. This implies, $a_2$ and $a_2''$ are present at the node $u$, and $a_2'$ is present at the node $v$. In this scenario, $G_1$ may visit either of the nodes $u$ or $v$.
     \begin{itemize}
         \item[(2a)] If $G_1$ visits $u$ and it has to move through edge $(u,v)$ that is missing. The group $G_1$ can not deviate from its path, and $G_2$ needs its first helper in order to deviate ($G_2$ can not deviate unless it has both its helpers with it). In this case, $G_2$ can take one helper from $G_1$ and skip the current edge to proceed further. Precisely, the first helper of $G_1$, i.e. $a_1'$, becomes the first helper of $G_2$ and vice versa. $G_1$ understands that its (new) first helper (i.e., old $a_2'$) is at the adjacent node and $a_1''$ has to proceed with its movement through the edge $(u,v)$ until it is successful in doing so. The group $G_2$ can now skip this edge $(u,v)$ and continue its traversal further. Figure \ref{fig:grpchange1} represents this case.
         \item[(2b)] If $G_1$ visits the node $v$ and it has to move through edge $(v,u)$ that is missing, it sees the first helper of $G_2$ (i.e., $a_2'$) present here. The group $G_1$ understands that at the adjacent node $u$, $a_2$ and $a_2''$ are stuck due to the missing edge. In this scenario, entire $G_1$ becomes $G_2$ i.e., (old) $a_1$, $a_1'$, $a_1''$ becomes the (new) $a_2$, $a_2'$, $a_2''$ respectively. Thus, $G_2$ continues with its traversal further. While (old) $a_2'$ becomes the leader of $G_1$ (i.e., (new) $a_1$) and attempts to move through the edge $(v,u)$ in the subsequent rounds. When (new) $a_1$ is successful in reaching at $u$, it meets with (old) $a_2$ which now becomes (new) $a_1'$. The (old) $a_2''$ moves to $v$. Now, the (new) $a_1$ and $a_1'$ again moves through edge $(u,v)$ to re-unite with (old) $a_2''$ and make it the new second helper of $G_1$. In this way, all three agents of (new) $G_1$ re-unite and proceed further according to their traversal strategy. Figure \ref{fig:grpchange2} represents this case.
     \end{itemize}
     \item The second helper of $G_2$ ($a_2''$) moves through port $p$ while the first helper of $G_2$ ($a_2'$) returns to $u$ at some round and then at the beginning of the next round, the edge $(u,v)$ is deleted by the adversary. This case is similar to case $2$. Again $G_1$ may visits $u$ or $v$ during its traversal. The change of roles is similar to that of case $2$ ensuring that $G_2$ can continue its traversal without waiting for the missing edge to appear. It is important to note that if $G_1$ does not visit any of the nodes $u$ or $v$, then $G_2$ can not proceed if it is stuck due to a missing edge.
 \end{itemize}
     
 This algorithm solves the problem of $1$-BHS. Now we proceed with the correctness and analysis of our algorithm.

\begin{figure}[ht!]
\centering
  \includegraphics[width=.7\linewidth]{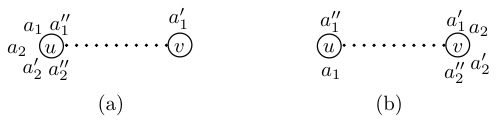} 
  \caption{This figure depicts the Case $1$ of our algorithm. In this case, the leader ($a_1$) and the second helper ($a_1''$) of $G_1$ are present at $u$ while the first helper ($a_1'$) of $G_1$ is present at $v$. The edge $(u,v)$ is missing. All the agents of $G_2$ are together, and when they visit either of the nodes $u$ or $v$ and have to move via the edge $(u,v)$, they skip the edge without the need for group change. Furthermore, for the sake of simplicity, we use $l^{G_1}$, $h_1^{G_1}$, and $h_2^{G_1}$ to denote the leader, first helper, and second helper, respectively, of group $1$. Similarly, $l^{G_2}$, $h_1^{G_2}$, and $h_2^{G_2}$ to denote the leader, first helper, and second helper, respectively, of group $2$.  }\label{fig:nochange}
\end{figure}

\begin{figure}[ht!]
\centering
  \includegraphics[width=.9\linewidth]{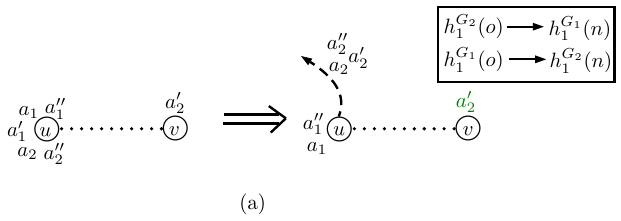} 
  \caption{This figure depicts the Case $(2a)$ of our algorithm. In this case, a change of group happens. When $G_1$ visits $u$ and meets with the $l^{G_2}=a_2$ and $h_2^{G_2}=a_2''$ ($h_1^{G_2}=a_2'$ has moved to $v$ while performing cautious movement and is stuck due to the missing edge $(u,v)$), agents of $G_2$ at $u$ combines with the first helper from $G_1$ (which now becomes the first helper of $G_2$) and proceeds by skipping the edge $(u,v)$. Agents $l^{G_1}$ and $h_2^{G_1}$ waits at $u$ for the missing edge to reappear. The agent at $v$, i.e., $a_2'$ is denoted in green. This represents it is now the first helper of $G_1$ but does not know it yet. Here two agents change their role. The old first helper of $G_2$ ($h_1^{G_2}(o)$) becomes the new first helper of $G_1$ ($h_1^{G_1}(n)$), and vice versa.}\label{fig:grpchange1}
\end{figure}

\begin{figure}[ht!]
\centering
  \includegraphics[width=.9\linewidth]{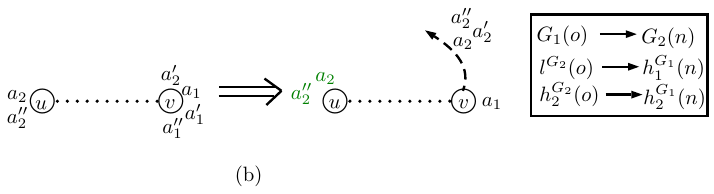} 
  \caption{This figure depicts the Case $(2b)$ of our algorithm. In this case, a change of group happens. When $G_1$ visits $v$ and meets with $h_1^{G_2}=a_2'$, the group $G_1$ now becomes group $G_2$ and $h_1^{G_2}=a_2'$ becomes the new leader of $G_1$ i.e., $a_1$. Agents $l^{G_2}$ and $h_2^{G_2}$ becomes the new first and second helpers of $G_1$ and are denoted in green as they do not know that the change of group has occurred. The agent $a_1$ waits at $v$ for the missing edge $(u,v)$ to reappear. Note that in this case, $a_1$ understands that the adjacent node $u$ is a safe node. Thus, it attempts to move through the edge $(u,v)$ to reach $u$ where its two new helpers are waiting.} \label{fig:grpchange2}
\end{figure}

\subsection{Correctness and Analysis}

\begin{observation}
    Although two different leaders of $G_2$ may exist at a given time, it is easy to observe that the new leader can distinguish itself from the old leader upon meeting, and the execution of $G_2$'s algorithm is continued only by the new leader and not by the old leader of $G_2$.
\end{observation}
\begin{theorem}\label{thm:maincorrectness}
    At least one of the groups $G_1$ or $G_2$ explores the graph $G$ correctly in $O(|E|^2)$ rounds.
\end{theorem}

\begin{proof}
    The algorithm partitions the agents into two groups, $G_1$ and $G_2$, each consisting of three agents. These groups are designed to replicate the behavior of agents $A_1$ and $A_2$, respectively, as described in the exploration strategy. According to Lemma \ref{thm:expwith2}, agents $A_1$ and $A_2$ correctly explore the dynamic graph. Since $G_1$ and $G_2$ operate using cautious walk, any action such as skipping a missing edge or initiating a new DFS traversal requires all three members of a group to be together. In Case 1 of the algorithm, if $G_2$ remains intact while $G_1$ is blocked by a missing edge $(x, y)$, and $G_2$ subsequently visits node $x$ or $y$, no group change is necessary; $G_1$ and $G_2$ continue their traversal and the movement by $G_1$ and $G_2$ is same as that of $A_1$ and $A_2$, respectively. In Case $2(a)$, if the first helper of $G_2$ is stuck at $y$ and the leader and second helper are at $x$, and $G_1$ reaches $x$, a group change occurs. This allows $G_2$ to skip the edge $(x, y)$ and continue its exploration, while $G_1$ waits at $x$ for the edge to reappear. In this scenario, the effective movements of $G_1$ and $G_2$ still align with those of $A_1$ and $A_2$. In Case $2(b)$, if $G_1$ visits node $y$, a new $G_2$ is formed at $y$ that comprises of all the three (old) members of $G_1$. $G_2$ continues its traversal. And the (old) first helper of $G_2$ becomes the (new) leader of $G_1$. In this case, edge $(x,y)$ is not skipped by $G_2$. Although this case leads to a deviation from the original paths of $A_1$ and $A_2$, the result is preferable since the edge is not skipped and $G_2$ is no longer blocked. As Theorem \ref{thm:expwith2} guarantees that at least one of $A_1$ and $A_2$ visit all nodes in the graph in $O(|E|^2)$ rounds, and since $G_1$ and $G_2$ simulate their behavior through cautious walk, every node is visited by at least one of the groups. This proves that at least one of $G_1$ or $G_2$ explores the graph $\mathcal{G}$ in $O(|E|^2)$ rounds.
\end{proof}

\begin{theorem}
    The problem of $1$-BHS can be solved by $6$ agents starting from a rooted initial configuration in $O(|E|^2)$ rounds where each agent has $O(\log n)$ bits of storage. Furthermore, each node requires only $O(\log n)$ bits of storage.
\end{theorem}
\begin{proof}
    The groups $G_1$ and $G_2$ replicates the movement strategy of $A_1$ and $A_2$ (as per the exploration strategy) by changing each edge movement with the cautious movement strategy. According to Theorem \ref{thm:maincorrectness}, at least one of $G_1$ and $G_2$ explores $G$ in $O(|E|^2)$ rounds. Since both the groups maintain cautious walk throughout the algorithm, black hole is detected either by $G_1$ or $G_2$ in $O(|E|^2)$ rounds.

    In Theorem \ref{thm:1dynamic}, we show that each agent utilizes $O(\log n)$ bits of memory to store all the parameters. In our algorithm, agents need to remember the leader and the two helpers that are a part of their group. This can also be accomplished using $O(\log n)$ memory to solve 1-BHS.

    At each node $v$, agent $a_i$ writes the $parent$ information corresponding to its DFS traversal. This requires $O(\log \Delta)$ storage, where $\Delta$ is the maximum degree of a node in graph $G$. Further, $label$ parameter written on the whiteboard by group $G_2$ is bounded by $O(|E|^2)$, and thus, $label$ parameter can be maintained on the whiteboard using $O(\log n)$ bits. Hence, our algorithm requires $O(\log n)$ storage at each node of the graph. This completes the proof.
\end{proof}

\section{Conclusion}
We propose an algorithm for the $1$-BHS problem using $6$ agents and another algorithm for the $f$-BHS problem using $6f$ agents. 
Both the problems are open to be studied from arbitrary initial configuration. 
Studying the problem of knowing the exact location of the black hole on arbitrary graphs or even on any graph class, i.e., knowing all the edges leading to the black hole can also be studied.


\bibliography{Bib}

\end{document}